\documentclass{llncs}
\newif\iffinal
\newif\ifwithextensions
\withextensionsfalse

\usepackage{booktabs}
\usepackage{amsfonts}
\usepackage{amsmath}
\usepackage{temporal}
\usepackage{cite}
\usepackage{graphicx}
\usepackage{enumitem}
\setitemize{noitemsep,topsep=2pt,parsep=0pt,partopsep=2pt}
\usepackage{wrapfig}
\usepackage{environ}
\usepackage{xspace}
\usepackage{color,soul}
  \definecolor{lightblue}{rgb}{.8,.95,1}

\usepackage{tikz}
\usepackage{pgf}
\usetikzlibrary{arrows,decorations,positioning,calc,automata,arrows,fit,shapes,backgrounds,petri}
\usepackage{caption}
\usepackage{mdframed}

\usepackage{scalerel}

\usepackage{todonotes}

\newcommand{\mailto}[1]{\href{mailto:#1}{\nolinkurl{#1}}}

\newcommand{\mA}{\mathcal{A}}
\newcommand{\mP}{\mathcal{P}}
\newcommand{\mG}{\mathcal{G}}

\newcommand{\guardset}{\mathcal{G}}

\newcommand{\init}{{\sf init}\xspace}
\newcommand{\templates}{P}

\newcommand{\bfair}[1]{\ensuremath{{#1}\text{-gfair}}}
\newcommand{\lbfair}[1]{\ensuremath{{#1}\text{-lfair}}}
\newcommand{\ufair}{\ensuremath{u\text{-fair}}}

\newcommand{\visInf}[1]{\mathsf{Visited}^\inf_{#1}}
\newcommand{\visFin}[1]{\mathsf{Visited}^\fin_{#1}}

\newcommand{\first}{\mathsf{first}}

\newcommand{\guardstates}{G}

\definecolor{darkgreen}{rgb}{0,0.5,0}
\definecolor{darkblue}{rgb}{0,0,.5}
\definecolor{mygray}{gray}{.3}

\newcommand{\state}{q}

\newcommand{\stateset}{\expandafter\MakeUppercase\expandafter{\state}}
\newcommand{\State}{s}

\newcommand{\Stateset}{\expandafter\MakeUppercase\expandafter{\State}}

\newcommand{\trans}{\ensuremath{\delta}}
\newcommand{\Trans}{\ensuremath{\Delta}}

\newcommand{\appears}{{\sf appears}\xspace}
\newcommand{\last}{{\sf last}\xspace}

\newcommand{\Nat}{\ensuremath{\mathbb{N}}}
\newcommand{\cupdot}{\mathbin{\dot{\cup}}}

\newcommand{\smartpar}[1]{\medskip \noindent {\bf #1}}

\renewcommand{\iff}{\Leftrightarrow}

\newtheorem{cor}{Corollary}
\newcommand{\cutoffsys}{\ensuremath{A {\parallel} B^c}\xspace}
\newcommand{\largesys}{\ensuremath{A {\parallel} B^n}\xspace}
\newcommand{\largesyse}[1]{\ensuremath{A {\parallel} B^{#1}}\xspace}

\newcommand{\fin}{\textit{fin}}

\renewcommand{\inf}{\textit{inf}}

\iffinal
\newcommand{\gray}[1]{}

\else
\newcommand{\gray}[1]{{\color{black!50} #1}}

\fi

\newcommand{\li}{\begin{itemize}}
\newcommand{\il}{\end{itemize}}

\NewEnviron{tikzLTS}{%
  \begin{tikzpicture}[node distance=2cm,inner sep=1pt,minimum size=0.5mm,bend angle=20]
	 	\tikzstyle{proc} = [rectangle,draw=black,fill=green!20,thick,inner sep=10pt]
	  	\tikzstyle{state} = [circle,draw=black,thick,inner sep=3pt]
	  	\tikzstyle{noproc} = [circle]
	  	\tikzstyle{lbl} = [rectangle,node distance=2cm]
  		\tikzstyle{pre} = [ <-,shorten <=2pt,shorten >=2pt, >=stealth', semithick]
	  	\tikzstyle{post} = [ ->,shorten <=2pt,shorten >=2pt, >=stealth', semithick]
    \BODY
  \end{tikzpicture}
}

\NewEnviron{tikzPath}{%
  \begin{tikzpicture}	
  	\tikzstyle{pstate} = [circle, fill=black,thick, inner sep=2pt, minimum size=2mm]
    \tikzstyle{hiddenstate} = [circle]
  	\tikzstyle{trans-notsched} = [ ->,shorten <=2pt,shorten >=2pt, >=stealth', semithick]
  	\tikzstyle{trans-sched} = [ ->,shorten <=2pt,shorten >=2pt, >=stealth', very thick] 
    \BODY
    \end{tikzpicture}
}

\newcommand{\U}{\textbf{U}\xspace}
\newcommand{\R}{\textbf{R}\xspace}
\newcommand{\F}{\textbf{F}\xspace}
\newcommand{\G}{\textbf{G}\xspace}
\newcommand{\X}{\textbf{X}\xspace}
\newcommand{\p}{\textbf{p}\xspace}
\newcommand{\buchi}[1]{\ensuremath{A}_{#1}}

\newcommand\LTLmX{\ensuremath{\mbox{\text{LTL}}{\setminus}\X}\xspace}
\newcommand\LTL{\ensuremath{\mbox{\text{LTL}}}\xspace}

\newcommand{\prompt}{\ensuremath{\mbox{\text{Prompt-LTL}}}\xspace}
\newcommand{\Prompt}{\prompt}
\newcommand{\promptmX}{\ensuremath{\prompt{\setminus}\X}\xspace}
\newcommand{\PromptmX}{\promptmX}

\newcommand{\qt}{Q_{\scaleto{T}{3pt}}}
\newcommand{\initt}{I_{\scaleto{T}{3pt}}}
\newcommand{\deltat}[1]{\delta_{\scaleto{#1}{4pt}}}
\newcommand{\sigt}{\Sigma_{\scaleto{T}{3pt}}}

\newcommand{\gbmodels}{\models_{gb}}
\newcommand{\ngbmodels}{\not \models_{gb}}
\newcommand{\lbmodels}{\models_{lb}}
\newcommand{\nlbmodels}{\not \models_{lb}}

\newcommand{\fairSys}[2]{\ensuremath{A {\parallel_{\scaleto{#1}{3pt}}} B^{#2}}\xspace}

\usepackage[colorlinks=true]{hyperref}
\hypersetup{%
  pdftitle={Promptness and Bounded Fairness in Concurrent and Parameterized Systems},%
  pdfauthor={Swen Jacobs, Mouhammad Sakr, Martin Zimmermann},%
  pdfsubject={},%
  pdfkeywords={},%
  colorlinks%
}

\title{Promptness and Bounded Fairness in Concurrent and Parameterized Systems \thanks{Partially funded by grant EP/S032207/1 from the Engineering and Physical Sciences Research Council (EPSRC).} 
}

\author{Swen Jacobs\inst{1} \and Mouhammad Sakr\inst{1,2} \and Martin Zimmermann\inst{3}}

\institute{CISPA Helmholtz Center for Information Security, Saarbr\"ucken, Germany \and
Saarland University, Saarbr\"ucken, Germany \and
University of Liverpool, Liverpool, United Kingdom}
\date{}
\begin{document}

\maketitle

\begin{abstract}
We investigate the satisfaction of specifications in Prompt Linear Temporal 
Logic (\prompt) by concurrent systems. Prompt-LTL is an extension of LTL that 
allows to specify parametric bounds on the satisfaction of eventualities, 
thus adding a quantitative aspect to the specification language. We 
establish a connection between bounded fairness, bounded stutter equivalence, 
and the satisfaction of \promptmX formulas. 
Based on this connection, we prove the first cutoff results for different 
classes of systems with a parametric number of components and quantitative 
specifications, thereby identifying previously unknown decidable fragments of 
the parameterized model checking problem. 
%Our results are obtained by generalizing the existing proof methods from \LTLmX to \promptmX. 
%Moreover, we show how to compute bounds on eventualities for arbitrary instances of the parameterized system from the bounds in the cutoff system. This immediately provides smaller worst-case bounds for systems of arbitrary size, and allows us to further reduce these bounds by reasoning only about the cutoff system. 

\end{abstract}

% !TEX root =  paper.tex
\section{Introduction}
\label{sec:intro}

% Motivation: (synthesis is especially valueable in concurrent settings)
Concurrent systems are notoriously hard to get correct, and are therefore a promising application area for formal methods like model checking or synthesis. However, these methods usually give correctness guarantees only for systems with a given, fixed number of components, and the state explosion 
problem prevents us from using them for systems with a large number 
of components. To ensure that desired properties hold 
for systems with a very large or even an \emph{arbitrary} number of components, methods for 
\emph{parameterized} model checking and synthesis have been devised.

While parameterized model checking is undecidable even for simple 
safety properties and systems with uniform finite-state 
components~\cite{Suzuki88}, there exist a number of methods that decide the 
problem for specific classes of 
systems~\cite{German92,EsparzaFM99,EmersonK00,Emerso03,EmersonK03,Clarke04c,AJKR14,Namjoshi07,EsparzaGM16}, 
some of which have been collected in surveys of the literature 
recently~\cite{Esparza14,BloemETAL15}. Additionally, there are semi-decision 
procedures that are successful in many interesting 
cases~\cite{Kurshan95,Bouajjani00,Clarke08,KaiserKW10,PnueliRZ01}.
However, most of these approaches only support safety properties, or their 
support for progress or liveness properties is limited, e.g., because global 
fairness properties are not considered and cannot be expressed in the 
supported logic (cp. Au{\ss}erlechner et al.~\cite{AJK16}).

In this paper, we investigate cases in which we can guarantee that a system 
with an arbitrary number of components satisfies strong liveness 
properties, including a quantitative version of liveness called 
\emph{promptness}. The idea of promptness is that a desired event should not 
only happen at \emph{some} time in the future, but there should exist a 
\emph{bound} 
on the time that can pass before it happens.
We consider specifications in \Prompt, an extension of \LTL 
with an operator that expresses prompt eventualities~\cite{KupfermanPV09},
 i.e., the logic puts a 
symbolic bound on the satisfaction of the eventuality, and the model 
checking problem asks if there is a value for this symbolic bound such that 
the property is guaranteed to be satisfied with respect to this value.
In many settings, adding promptness comes for free in terms of asymptotic 
complexity~\cite{KupfermanPV09}, e.g., model checking and synthesis~\cite{JTZ2018}.\footnote{\Prompt can be seen as a fragment of parametric LTL, a logic introduced by Alur et al.~\cite{AlurETP01}. However, since most decision problems for parametric LTL, including model checking, can be reduced to those for \Prompt, we can restrict our attention to the simpler logic.}  
Hence, here we study 
\emph{parameterized} model checking for \Prompt and show that in many cases 
adding promptness is also free for this problem. 

More precisely, as is common in the analysis of concurrent systems, we 
abstract concurrency by an interleaving semantics and consider the 
satisfaction of a specification \emph{up to stuttering}. Therefore, we limit 
our specifications to \PromptmX, an extension of the stutter-insensitive 
logic \LTLmX that does not have the next-time operator.
Determining satisfaction of \PromptmX specifications by concurrent systems 
brings new challenges and has not been investigated in detail before.

\paragraph{\bf Motivating Example.}

\begin{wrapfigure}{r}{0.45\linewidth}
%\iflipics
%\vspace{-18pt}
%\else
%\vspace{-20pt}
%\hspace{-10pt}
%\fi
\scalebox{0.95}{
\begin{tikzpicture}[node distance=1.7cm,>=stealth,auto]
  \tikzstyle{state}=[circle,thick,draw=black,minimum size=7mm]

  \begin{scope}
    
    \node [state] (init) {$\init$};   
    
    \node [state] (r) [left of=init] {$\mathbf{r}$}     
     edge [post] (init);
     
    \node [state] (tr) [below of=r] {$\mathbf{tr}$}     
     edge [pre] (init)
     edge [post] node[left = 0.5pt] {\tiny{$\forall \mathbf{\neg w}$}} (r);
     
    \node [state] (w) [right of=init] {$\mathbf{w}$}     
     edge [post] (init);    
     
    \node [state] (tw) [below of=w] {$\mathbf{tw}$}     
     edge [pre] (init)
     edge [post] node[right =0.5pt] {\tiny{$\forall \mathbf{\neg \{w, r\}}$}}(w);

  \end{scope}
  \end{tikzpicture}
}
\label{fig:reader-writer}
%\iflipics
%\vspace{-28pt}
%\else
%\vspace{-20pt}
%\hspace{-10pt}
%\fi
\end{wrapfigure}
For instance, consider the reader-writer protocol on the right which models 
access to shared data between processes. If a process wants to ``read'', it 
enters the state $\mathbf{tr}$ (``try-read'') that has a direct transition to the 
reading state $\mathbf{r}$. However, this transition is guarded by $\forall \neg \mathbf{w}$, 
which stands for the set of all states except $\mathbf{w}$, the ``writing'' state. 
That is, the transition is only enabled if no other process is currently in 
state $\mathbf{w}$. Likewise, if a process wants to enter $\mathbf{w}$ it has to go through 
$\mathbf{tw}$, but the transition to $\mathbf{w}$ is enabled only if no other process is 
reading or writing.

For such systems, previous results~\cite{EmersonK00,AJK16} provide cutoff results for parameterized verification of properties from \LTLmX, e.g., 
$$\forall i. \G \left( (tr_i \rightarrow \F r_i) \land (tw_i \rightarrow \F w_i) \right),$$
In this paper we investigate whether the same cutoffs still hold if we consider 
specifications in \PromptmX, e.g., if we substitute the \LTL eventually 
operator~$\F$ above with the prompt-eventually operator $\F_{\p}$, while imposing a 
bounded fairness assumption on the scheduler.

\paragraph{\bf Contributions.}
As a first step, we note that \PromptmX is not a stutter-insensitive logic, 
since unbounded stuttering could invalidate a promptness property. This leads 
us to define the notion of \emph{bounded stutter equivalence}, and proving 
that \PromptmX is \emph{bounded stutter insensitive}.

This observation is then used in an investigation of existing approaches that 
solve parameterized model checking by the \emph{cutoff} method, which reduces 
problems from 
systems with an arbitrary number of components to systems with a fixed number 
of components. More precisely, these approaches prove that for every trace in 
a large system, a stutter-equivalent trace in the cutoff system exists, and 
vice versa. We show that in many cases, modifications of these constructions 
allow us to 
obtain traces that are \emph{bounded} stutter equivalent, and therefore the 
cutoff results extend to specifications in \PromptmX.
The types of systems for which we prove these results include \emph{guarded 
protocols}, as introduced by Emerson and Kahlon~\cite{EmersonK00}, and 
\emph{token-passing systems}, as introduced by Emerson and Namjoshi\cite{Emerso03} 
for uni-directional rings, and by Clarke et al.~\cite{Clarke04c} for arbitrary 
topologies. Parameterized model checking for both of these system classes has 
recently been further 
investigated~\cite{AJK16,JS18,AJKR14,AminofKRSV18,SpalazziS17,SpalazziS14}, 
but thus far not in a context that includes promptness properties.
%Conjunctive guards can 
%be used to model atomic sections or locks, while disjunctive guards can model 
%pairwise rendezvous or token-passing. 
%Due to space constraints, we omit detailed proofs in some cases. They can be found in the full version of the paper~\cite{JacobsSZ19}.

\section{\PromptmX and Bounded Stutter Equivalence}

We assume that the reader is aware of standard notions such as finite-state transition systems and linear temporal logic (LTL)~\cite{PrinciplesMC}. 

We consider concurrent systems that are represented as an interleaving composition of finite-state transition systems, possibly with synchronizing transitions where multiple processes take a step at the same time. In such systems, a process may stay in the same state for many global transitions while other processes are moving. From the perspective of that process, these are \emph{stuttering steps}. 

Stuttering is a well-known phenomenon, and temporal languages that include the next-time operator \X are \emph{stutter sensitive}: they can require some atomic proposition to hold at the next moment in time, and the insertion of a stuttering step may change whether the formula is satisfied or not. On the other hand, LTL$\setminus$\X, which does not have the \X operator, is stutter-insensitive: two words that only differ in stuttering steps cannot be distinguished by the logic~\cite{PrinciplesMC}.

In the following, we introduce \PromptmX, an extension of \LTLmX, and investigate its properties with respect to stuttering.

\subsection{\PromptmX}\label{promptLTL definition} Let $AP$ be the set of atomic propositions. The syntax of \PromptmX formulas over $AP$ is given by the following grammar:

$$\varphi ::= a \mid \neg a \mid \varphi \lor \varphi \mid \varphi \land \varphi \mid \F_{\p} \varphi \mid \varphi \U \varphi \mid \varphi \R \varphi, \text{ where } a\in AP$$

The semantics of \PromptmX formulas is defined over infinite words $w=w_0w_1\ldots \in (2^{AP})^\omega$, positions $i \in \Nat$, and bounds $k \in \Nat$. 
The semantics of the prompt-eventually operator $\F_{\p}$ is defined as follows: 
$$(w,i,k) \models \F_{\p}\varphi \text{ iff there exists } j \text{ such that } i \leq j \leq i + k \text{ and } (w,j,k) \models \varphi.$$ 
%We say $(w,i,k) \models \varphi$ if $\varphi$ holds at position $i$ of $w$ with bound $k$.
All other operators ignore the bound $k$ and have the same semantics as in LTL, moreover we define $\F$ and $\G$ in terms of $\U$ and $\R$ as usual.

%\begin{lemma}\label{gb-lb-fair-relation}
%Let $A$, $B$ be process templates, and a specification $\varphi(A,B_1,\ldots,B_c)$ with $\varphi \in$ \LTLmX $\cup$ \promptmX then, $\forall n \geq c$:
%$$ \left( \largesys \gbmodels \varphi(A,B_1,\ldots,B_c) \implies \largesys \lbmodels \varphi(A,B_1,\ldots,B_c) \right)$$
%\begin{proof}
%Assume that the implication is not true, then there is $b \in \Nat$, a run $x$ of $\largesys$ where $\lbfair{b}(x, \{A,B_1,\ldots,B_c\})$ with $x \not \models \varphi(A,B_1,\ldots,B_c)$, and $\forall y$ of $\largesys$ where $\gbfair{b'}(y)$ with $y \not \models $
%\end{proof}

%\end{lemma}

\subsection{\Prompt and Stuttering}

Our first observation is that \PromptmX is stutter sensitive: to satisfy the formula 
$\varphi = \G \F_{\p} q$ with respect to a bound $k$, $q$ has to appear at 
least once in every $k$ steps. Given a word $w$ that satisfies $\varphi$ for 
some bound $k$, we can construct a word that does not satisfy $\varphi$ for 
any bound $k$ by introducing an increasing (and unbounded) number of 
stuttering steps between every two appearances of $q$. In the following, we 
show that \PromptmX is stutter insensitive if and only if there is a bound on 
the number of consecutive stuttering steps.

%\begin{definition}[Bounded Stutter Equivalence]
%Two words $w',w_2 \in (2^{AP})^\omega$ are $d$-stutter equivalent, denoted by $w' \equiv_{d}^{s} w_2$, if there exist two infinite sequences of non-zero natural numbers $n_0,n_1,n_2,...$ and $m_0,m_1,m_2,...$   such that \ms{$s_i$ and $s_j$ might be equal}

%$$\begin{array}{ll}
%w'  = & \underset{n_0 \, times}{\underbrace{s_0 \ldots s_0}} \; \underset{n_1 \, times}{\underbrace{s_1\ldots s_1}} \; \underset{n_2 \, times}{\underbrace{s_2\ldots s_2}}\ldots\\
%w_2  = & \underset{m_0 \, times}{\underbrace{s_0\ldots s_0}} \; \underset{m_1 \, times}{\underbrace{s_1\ldots s_1}} \; \underset{m_2 \, times}{\underbrace{s_2\ldots s_2}}\ldots ,
%\end{array}$$

%and $n_i \leq  m_i.d$ $m_i \leq  n_i.d$ or  for all $i$.
%\end{definition}

\smartpar{Bounded Stutter Equivalence.}
A finite word $w \in (2^{AP})^+$ is a \emph{block} if $\exists \alpha \subseteq AP$ such that $w = \alpha^{|w|}$.
Two blocks $w, w' \in (2^{AP})^+$ are \emph{d-compatible} if $\exists \alpha \subseteq AP$ such that $w = \alpha^{|w|}, w' = \alpha^{|w'|}$, $|w| \leq d\cdot|w'|$ and $|w'| \leq d\cdot|w|$.
Two infinite sequences of blocks $w_0w_1w_2\ldots$, $w'_0w'_1w'_2\ldots$ are \emph{d-compatible} if $w_i, w'_i$ are d-compatible for all $i$.

Two words $w, w' \in (2^{AP})^\omega$ are \emph{$d$-stutter equivalent}, denoted $w \equiv_{d} w'$, if they can be written as $d$-compatible sequences of blocks. They are \emph{bounded stutter equivalent} if they are $d$-stutter equivalent for some $d$. We denote by $\hat{w}$ a sequence of blocks that corresponds to a word $w$.

Given an infinite sequence of blocks $\hat{w} = w_0,w_1,w_2\ldots$, let 
$N_i^{\hat{w}}=\{\sum_{l=0}^{i-1} |w_l|,$ $\ldots,\sum_{l=0}^{i-1} |w_l| + |w_i| - 1\}$ 
be the set of positions of the $i$th block. Given a position $n$, there is 
a unique $i$ such that $n \in N_i^{\hat{w}}$.

To prove that \PromptmX is \emph{bounded stutter insensitive}, i.e., it cannot distinguish two words that are bounded stutter equivalent, we define a function that maps between positions in two $d$-compatible sequences of blocks:
given two infinite $d$-stutter equivalent words $w, w'$ such that $\hat{w},\hat{w}'$ are $d$-compatible, define the function $f: \mathbb{N} \rightarrow 2^{\mathbb{N}}$ where: $f(j) = N_i^{\hat{w}'}$ $\iff$ $j \in N_i^{\hat{w}}$. Note that $\forall j' \in f(j)$ we have $w(j) = w'(j')$, where $w(i)$ denotes the $i$th symbol in $w$. For an infinite word $w$, let $w[i,\infty)$ denote its suffix starting at position $i$, and $w[i:j]$ its infix starting at $i$ and ending at $j$. 
Then we can state the following.

%Given a word $w$ with the structure 
%$$w  = \underset{n_0 \, times}{\underbrace{s_0\ldots s_0}} \; \underset{n_1 \, times}{\underbrace{s_1\ldots s_1}}\ldots \underset{n_i \, times}{\underbrace{s_i\ldots s_i}}\ldots, $$
%let $N_0^w=\{1,\ldots,n_0\}$, and $N_i^w=\{\sum_{l=0}^{i-1} n_l + 1,\ldots, \sum_{l=0}^{i-1} n_l + n_i\}$ for $i > 0$. 
%Note that for every $j,j' \in N_i^w$, we have $w(j)=w(j')$. \sj{notation $w(j)$ not defined}

%Based on these sets, we can relate bounded stutter equivalent words. For $w$ and $w'$ with $w \equiv_d^s w'$, define the function\sj{should the function have a name?} $f_{w,w'}: \mathbb{N} \rightarrow 2^{\mathbb{N}}$ where: $f(j) = N_i^{w'}$ $\iff$ $j \in N_i^{w}$. Note that $\forall j' \in f(i)$ we have $w(i) = w'(j')$. Based on this observation, we can state the following lemma.

\begin{remark}\label{remark-stutter}
Given two words $w$ and $w'$, if $w \equiv_d w'$, then $\forall j \in \mathbb{N}$ $\forall j' \in f(j): $ $w[j,\infty) \equiv_d w'[j',\infty)$.
\end{remark}

Now, we can state our first theorem.

\begin{theorem}[\PromptmX is Bounded Stutter Insensitive]\label{theorem-stutter}
Let $w, w'$ be $d$-stutter equivalent words, $\varphi$ a \PromptmX formula, and $f$ as defined above. Then $\forall i,k \in \Nat$:
$$ (w,i,k) \models \varphi ~\Rightarrow~ \forall j \in f(i): (w',j,d \cdot k) \models \varphi$$
\end{theorem}

\begin{proof}

The proof works inductively over the structure of $\varphi$. Let $\hat{w}=w_0,w_1,\ldots$ and $\hat{w}'=w'_0,w'_1,\ldots$ be two d-compatible sequences of $w$ and $w'$. We denote by $n_i,m_i$ the number of elements inside $N_i^{w},N_i^{w'}$ respectively. We consider two cases here, the full proof can be found in Appendix~\ref{app:proof1}.\\

\textbf{Case 1: $\varphi = \F_{\p} \varphi$.}
$(w,i,k) \models \F_{\p} \varphi$ $\iff \exists e,x: i \leq e \leq i+ k,$ 
$e \in N_x^{\hat{w}}$, and $(w,e,k) \models \varphi$ where 
$(\sum_{l=0}^{x-1} n_l) \leq e < (\sum_{l=0}^{x} n_l)$. Then by induction 
hypothesis we have: $\forall j \in f(e)$ $(w',j,d \cdot k) \models \varphi$. 
Let $s$ be the 
smallest position in $f(e)$, then $s = \sum_{l=0}^{x-1} m_l$. There exists $y 
\in \Nat$ s.t. $i \in N_y^{\hat{w}}$ then $s = \sum_{l=0}^{y-1} m_l + \sum_{l=y}^{x-1}
 m_l$ $\leq \sum_{l=0}^{y-1} m_l + \sum_{l=y}^{x-1} n_l.d \leq \sum_{l=0}^{y-1} 
m_l + d.(\sum_{l=y}^{x-1} n_l) $ $\leq \sum_{l=0}^{y-1} m_l + k \cdot d$ 
(note that $i \in N_y^{\hat{w}}$ and $(w,i,k) \models \F_{\p} \varphi$). As $\sum_{l=0}
^{y-1} m_l$ is the smallest position in $f(i)$, then $\forall j \in f(i):$ 
$(w',j,d \cdot k) \models \F_{\p} \varphi$.\\
%However $\forall l$ $m_l \leq n_l.d$.

\textbf{Case 2: $\varphi = \varphi_1 \U \varphi_2$.}
$(w,i,k) \models \varphi_1 \U \varphi_2$ $\iff \exists j \geq i:$ $(w,j,k) \models \varphi_2$ and $\forall e < j:$ $(w,e,k) \models \varphi_1$. Then, by induction hypothesis we have: $\forall e < j$ $\forall l \in f(e): (w', l, d \cdot k) \models \varphi_1$ and $\forall l \in f(j): (w', l, d \cdot k) \models \varphi_2$, therefore $\forall j \in f(i):$ $(w',j,d \cdot k) \models \varphi_1 \U \varphi_2$.
\qed
\end{proof}

Our later proofs will be based on the existence of counterexamples to a given property, and will use the following consequence of Theorem~\ref{theorem-stutter}. 

\begin{cor}\label{cor-stutter}
Let $w, w'$ be $d$-stutter equivalent words, $\varphi$ a \PromptmX formula, and $f$ as defined above. Then $\forall k \in \Nat$:
\begin{center}
$(w,i,k) \not \models \varphi ~\Rightarrow~ \forall j \in f(i):$ $(w',j,k/d) \not \models \varphi$  
\end{center}
\end{cor}

\section{Guarded Protocols and Parameterized Model Checking}

In the following, we introduce a system model for concurrent systems, called guarded protocols. However, we will see that some of our results are of interest for other classes of concurrent and parameterized systems, e.g., the token-passing systems that we investigate in Section~\ref{sec:TPS}.

\subsection{System Model: Guarded Protocols}
\label{sec:model}

We consider systems of the form $\largesys$, consisting of one copy of a process template $A$ and $n$ copies of a process template $B$, in an interleaving parallel composition.
We distinguish objects that belong to different templates by indexing them with the template. E.g., for process template $U \in \{A,B\}$, $Q_U$ is the set of states of $U$.
For this section, fix a finite set of states $\stateset = Q_A \cupdot Q_B$ and a positive integer $n$, and let $\guardset = \{\exists, \forall\} \times 2^Q$ be the set of guards.

%\footnote{Process template $A$ may be a trivial process that does nothing if we want to just consider a system $B^n$, as in the example in Section~\ref{sec:intro}.}%
%% AK: moved to a separate note, to be able to explain why we _cannot_ generalize for 1-conj
%\footnote{As shown in \cite{Emerson00}, cutoffs for this case generalize to cutoffs 
%          for systems of the form $A^m {\parallel} B^n$, and further 
%          to systems with an arbitrary number of process templates 
%          $U_1^{n_1} {\parallel} \ldots {\parallel} U_m^{n_m}$.} 

\smartpar{Processes.} A \emph{process template} 
 is a transition system
  $U=(\stateset_U, \init_U, \trans_U)$ where
	\begin{itemize}
	\item $\stateset_U \subseteq \stateset$ is a finite set of states including the
  initial state $\init_U$,
%	\item $\inputs$ is a finite input alphabet,\sj{simplify definitions by removing inputs and stating that everything extends to process templates that have additional local inputs in the same way as in AJK?} 
	\item $\trans_U \subseteq \stateset_U \times \guardset \times \stateset_U$ is a guarded transition relation.
	\end{itemize}
	%A process template is \emph{closed} if $\inputs = \emptyset$, and otherwise \emph{open}.
%
%Define the size of $U$ as $\card{U} = \card{\stateset_U}$. 
%An instance of process template $U$ will be called a \emph{$U$-process}.
%
\smartpar{Guarded Protocols.} 
The semantics of $\largesys$ is given by the transition system $(S, \init_s, \Delta)$, where
\footnote{By similar arguments as in Emerson and Kahlon~\cite{EmersonK00}, our results can be extended to systems with an arbitrary (but fixed) number of process templates. The same holds for \emph{open} process templates that can receive inputs from an environment, as considered by Au{\ss}erlechner et al.~\cite{AJK16}.}

\begin{itemize}
\item $S = Q_A \times (Q_B)^n$ is the set of (global) states,
\item $\init_S = (\init_A, \init_B,\ldots, \init_B)$ is the global initial state, and
\item $\Delta \subseteq S \times S$ is the global transition relation. $\Delta$ will be defined by local guarded transitions of the process templates $A$ and $B$ in the following.
\end{itemize}

We distinguish different copies of process template $B$ in $\largesys$ by subscript, and each $B_i$ is called a \emph{$B$-process}. We denote the set $\{A, B_1,\ldots,B_n\}$ as $\mP$, and a process in $\mP$ as $p$. For a global state $s \in S$ and $p \in \mP$, let the \emph{local state of $p$ in $s$} be the projection of $s$ onto that process, denoted $s(p)$.

Then a local transition $(q,g,q')$ of process $p \in \mP$ is \emph{enabled} in global state $s$ if $s(p)=q$ and either
\begin{itemize} 
\item $g=(\exists,\guardstates)$ and $\exists p' \in \mP \setminus \{p\}: s(p') \in \guardstates$, or
\item $g=(\forall,\guardstates)$ and $\forall p' \in \mP \setminus \{p\}: s(p') \in \guardstates$.
\end{itemize}

Finally, $(s,s') \in \Delta$ if there exists $p \in \mP$ such that $(s(p),g,s'(p)) \in \delta_p$ is enabled in $s$, and $s(p')=s'(p')$ for all $p' \in \mP \setminus \{p\}$. We say that the transition $(s,s')$ \emph{is based} on the local transition $(s(p),g,s'(p))$ of $p$.

%For $U \in \{A,B\}$, we write $G_U$ for the set of non-trivial guards that are used in $\trans_U$, i.e., guards different from $Q$ and $\emptyset$. Then, let $G = G_A \cup G_B$.

\smartpar{Disjunctive and Conjunctive Systems.}
We distinguish disjunctive and conjunctive systems, as defined by Emerson and 
Kahlon~\cite{EmersonK00}. In a \emph{disjunctive process template}, every guard is 
of the form $(\exists,\guardstates)$ for some $\guardstates \subseteq Q$. In 
a \emph{conjunctive process template}, every guard is of the form 
$(\forall,\guardstates)$, and $\{\init_A, \init_B\} \subseteq 
\guardstates$, i.e., initial states act as neutral states for all 
transitions. A \emph{disjunctive (conjuctive) system} consists of only disjunctive (conjunctive) process templates. For conjunctive systems we additionally assume that processes are \emph{initializing}, i.e., any process 
that moves infinitely often visits its initial state infinitely often.\footnote{This restriction has already been considered by Au{\ss}erlechner et
al.~\cite{AJK16}, and was necessary to support global fairness assumptions.}

%\smartpar{Disjunctive and Conjunctive Systems.}
%We distinguish disjunctive and conjunctive systems, as defined by Emerson and 
%Kahlon~\cite{EmersonK00}. In a \emph{disjunctive system}, every guard $g$ is 
%of the form $(\exists,\guardstates)$ for some $\guardstates \subseteq Q$. In 
%a \emph{conjunctive system}, every guard is of the form 
%$(\forall,\guardstates)$ with $\guardstates \subseteq Q$ and $\{\init_A, \init_B\} \subseteq 
%\guardstates$, i.e., initial states act as neutral states for all 
%transitions. For conjunctive systems we additionally assume that processes are \emph{initializing}, i.e., any process 
%that moves infinitely often visits its initial state infinitely often.\footnote{This restriction has already been considered by Au{\ss}erlechner et
%al.~\cite{AJK16}, and was necessary to support global fairness assumptions.}

\smartpar{Runs.}
A \emph{path} of a system $\largesys$ is a sequence $x = s_0 s_1 \ldots$ of global states such that for all $i < |x|$ there is a 
transition $(s_i,s_{i+1}) \in \Trans$ based on a local 
transition of some process $p \in \mP$. We say that 
$p$ \emph{moves} at \emph{moment} $i$. 
A path can be finite or infinite, and a \emph{maximal path} is a path that cannot be extended, i.e., it is either infinite or ends in a global state where no local transition is enabled, also called a \emph{deadlock}.
%Configuration $(s,e,\bot)$ appears
% iff all processes are disabled for $s$ and $e$.
%Also, for every $p$ and $\time < |x|$: 
%either $e_{\time+1}(p) = e_\time(p)$ or process $p$ moves at moment $\time$. 
%That is, the environment keeps input to each process unchanged until 
%the process can read it.\footnote{By only considering inputs that are actually %processed, we 
%approximate an 
%action-based semantics. Paths that do not fulfill this requirement are not 
%very interesting, since the environment can violate any interesting 
%specification that involves input signals by manipulating them when the 
%corresponding process is not allowed to move.} 
A \emph{run} is a maximal path starting in $\init_S$. We write $x \in \largesys$ to denote that $x$ is a run of $\largesys$.
%We say that a run is \emph{initializing} if every 
%%$B$-process 
%process
%that moves infinitely often also visits 
%%$\initstate_B$ 
%its initial state $\initstate$ 
%infinitely often.

Given a path $x = s_0 s_1 \ldots$ and a process $p$, the \emph{local path} of $p$ in $x$ is the projection $x(p) = s_0(p)s_1(p)\ldots$ of $x$ onto local states of $p$. It is a \emph{local run} of $p$ if $x$ is a run. Additionally we denote by $x(p_1,\ldots,p_k)$ the projection $s_0(p_1,\ldots,p_k) s_1(p_1,\ldots,p_k)\ldots$ of $x$ onto the processes $p_1,\ldots,p_k \in \mP$.
%Given a run $x$, a \emph{local run} $x(p) = (s_0(p), F_0)\;(s_1(p),F_1)\ldots$ is a \emph{local path} where we replace $p_i$ by a the so-called \emph{transition flag} $F$ which changes its polarity whenever the process moves i.e. $F_{i+1} \neq F_i$ if $p_{i}= p$ and $F_{i+1} = F_i$ otherwise. Additionally we denote by $x(p_i,\ldots,p_j)$ the projection $(s_0(p_i,\ldots,p_j), F_0)(s_1(p_i,\ldots,p_j), F_1)\ldots$ where $F$ changes its polarity whenever a process in $\{p_i,\ldots,p_j\}$ moves.\ms{I think there is no need for the introduced flag that indicates a local process progress}

\smartpar{Fairness.}
We say a process $p$ is \emph{enabled} in global state $s$ if at least one of its transitions is enabled in $s$, otherwise it is \emph{disabled}.
Then, an infinite run $x$ of a system $\largesys$ is
\begin{itemize}
\item \emph{strongly fair} if for every process $p$, if $p$ is enabled infinitely often, then $p$ moves infinitely often.
\item \emph{unconditionally fair}, denoted $\ufair(x)$, if every process moves infinitely often. 
\item \emph{globally $b$-bounded fair}, denoted $\bfair{b}(x)$, for some $b \in \Nat$, if 
\[ \forall p \in \mP\ \forall m \in \Nat\ \exists j \in \Nat: m \leq j \leq m+b \text{ and $p$ moves at moment } j.\]
\item \emph{locally $b$-bounded fair} for $E \subseteq \mP$, denoted $\lbfair{b}(x,E)$, if it is unconditionally fair and 
\[ \forall p \in E\ \forall m \in \Nat\ \exists j \in \Nat: m \leq j \leq m+b \text{ and $p$ moves at moment } j.\]
\end{itemize}

\smartpar{Bounded-fair System.} We consider systems that keep track of bounded fairness explicitly by running in parallel to $\largesys$ one counter for each process. In a step of the system where process $p$ moves, the counter of $p$ is reset, and all other counters are incremented. If one of the counters exceeds the bound $b$, the counter goes into a failure state from which no transition is enabled. 
We call such a system a \emph{bounded-fair system}, and denote it $\fairSys{b}{n}$.

A \emph{path} of a bounded-fair system $\fairSys{b}{n}$ is given as $x = (s_0,b_0)(s_1,b_1)\ldots$, and extends a path of $\largesys$ by valuations $b_i \in \{0,\ldots,b\}^{n+1}$ of the counters. Note that a run (i.e., a maximal path) of $\fairSys{b}{n}$ is finite iff either it is deadlocked (in which case also its projection to a run of $\largesys$ is deadlocked) or a failure state is reached. Thus, the projection of all infinite runs of $\fairSys{b}{n}$ to $\largesys$ are exactly the globally $b$-bounded fair runs of $\largesys$.

%A run is \emph{globally deadlocked} if it is finite.
%An infinite run is \emph{locally deadlocked} for process $p$ if there exists $\time$ such that $p$ is disabled for all $s_{\time'}$ with $\time'\ge \time$. A run is \emph{deadlocked} if it is locally or globally deadlocked.
%A system \emph{has a (local/global) deadlock} if it has a (locally/globally) deadlocked run. Note that absence of local deadlocks for all $p$ implies absence of global deadlocks, but not the other way around.

\subsection{Parameterized Model Checking and Cutoffs}
\label{sec:prelim}

\smartpar{\PromptmX Specifications.} 
Given a system $\largesys$, we consider specifications over $AP = Q_A \cup (Q_B \times \{1,\ldots,n\})$, i.e., states of processes are used as atomic propositions. For $i_1, \ldots, i_c \in \{1,\ldots,n\}$, we write $\varphi(A,B_{i_1},\ldots,B_{i_c})$ for a formula that contains only atomic propositions from $Q_A \cup (Q_B \times \{i_1,\ldots,i_c\})$.

In the absence of fairness considerations, we say that $\largesys$ \emph{satisfies }$\varphi$ if 
$$\exists k \in \Nat~ \forall x \in \largesys:~ (x,0,k) \models \varphi.$$

We say that $\largesys$ \emph{satisfies $\varphi(A,B_1,\ldots,B_c)$ under global bounded fairness}, written $\largesys \gbmodels \varphi(A,B_1,\ldots,B_c)$, if 
$$\forall b \in \Nat~ \exists k \in \Nat~ \forall x \in \largesys:~ \bfair{b}(x) \Rightarrow (x,0,k) \models \varphi(A,B_1,\ldots,B_c).$$ 

Finally, for local bounded fairness we usually require bounded fairness for all processes that appear in the formula $\varphi(A,B_1,\ldots,B_c)$. Thus, we say that \emph{$\largesys$ satisfies $\varphi(A,B_1,\ldots,B_c)$ under local bounded fairness}, written $\largesys \lbmodels \varphi(A,B_1,\ldots,B_c)$, if 
$$\forall b \in \Nat~ \exists k \in \Nat~ \forall x \in \largesys:~ \lbfair{b}(x,\{1,\ldots,c\}) \Rightarrow (x,0,k) \models \varphi(A,B_1,\ldots,B_c).$$

\smartpar{Parameterized Specifications.} 	
\label{sec:parameterized}
A \emph{parameterized specification} is a Prompt-\LTLmX formula
with quantification over the indices of atomic propositions. 
A \emph{$h$-indexed formula} is of the form $\forall{i_1},\ldots,\forall{i_h}. \varphi(A,B_{i_1},\ldots,B_{i_h})$. % or $\forall{i_1,\ldots,i_h.} \pexists \phi(A,B_{i_1},\ldots,B_{i_h})$. 
Let $f \in \{gb,lb\}$, then for given $n \geq h$, 
$$\largesys \models_f \forall{i_1},{\ldots},\forall{i_h}. \varphi(A,B_{i_1},{\ldots},B_{i_h})$$
\begin{center}
$\iff$
\end{center}
$$\forall j_1 \neq {\ldots} \neq j_h \in \{1,\ldots,n\}:~ \largesys \models_f \varphi(A,B_{j_1},{\ldots},B_{j_h}).$$ 
By symmetry of guarded protocols, this is equivalent 
(cp.\cite{EmersonK00})
to $\largesys \models_f \varphi(A,B_1,\ldots,B_h)$. 
The latter formula is denoted by $\varphi(A,B^{(h)})$, 
and we often use it instead of the original $\forall{i_1},\ldots,\forall{i_h}. \varphi(A,B_{i_1},...,B_{i_h})$. 
\smartpar{(Parameterized) Model Checking Problems.}
For $n \in \Nat$, a specification $\varphi(A,B^{(h)})$ with $n \ge h$, and $f \in \{gb,lb\}$:
\begin{itemize}
\item the \emph{model checking problem} is to decide whether  $\largesys$ $\models_f \varphi(A,B^{(h)})$,
%\item the (global/local) \emph{deadlock detection problem} is to decide whether $\largesys$
%      has (global/local) deadlocks,
%\item the \emph{deadlock detection problem} is to decide whether all runs of $\largesys$
%are infinite and $\largesys \models \pforall \spec_{\neg dead}$, 
%i.e., there are no local deadlocks,
\item the \emph{parameterized model checking problem} (PMCP) is to decide whether $\forall m \ge n: \largesyse{m} \models_f \varphi(A,B^{(h)})$.%, and 
%\item the \emph{parameterized deadlock detection problem} is to decide whether for all $m \ge n$, all runs of $(A,B)^{(1,m)}$ are infinite and $(A,B)^{(1,m)} \models \pforall \spec_{\neg dead}$.
%\item the \emph{parameterized (global/local) deadlock detection problem} is to decide whether 
%      for some $m \ge n$, $\largesyse{m}$ does have (global/local) local deadlocks.
\end{itemize}
%
%These definitions can be flavored with different notions of fairness, and with the $\pexists$ path quantifier instead of $\pforall$.
%\sj{removed remark about synthesis problems}
%Also, corresponding problems for the \emph{synthesis} of process templates can be defined (compare Au{\ss}erlechner et al.~\cite{AJK16}). Parameterized synthesis based on cutoffs~\cite{Jacobs14} is also supported by our cutoff results, but the details will not be necessary for understanding the results presented here.

%For a given number $n \in \bbN$ and specification $\phi(A,B^{(k)})$ with $n \ge k$,
%\begin{itemize}
%\item the \emph{template synthesis problem} is to find process templates $A,B$ such that
%$\largesys \models \pforall \phi(A,B^{(k)})$ and $\largesys$ does not have global deadlocks. 
%\item 
%the \emph{bounded template synthesis problem} for a pair of bounds $(\bound_A,\bound_B) \in \bbN \times \bbN$ 
%is to solve the template synthesis problem with 
%$\card{A} \leq \bound_A$ and $\card{B} \leq \bound_B$.
%\item the \emph{parameterized template synthesis problem} is to find process templates $A,B$ such that $\forall m \ge n:\ (A,B)^{(1,m)} \models \pforall \phi(A,B^{(k)})$ and $(A,B)^{(1,m)}$ does not have global deadlocks.
%\end{itemize}

\smartpar{Cutoffs and Decidability.}
We define cutoffs with respect to a class of systems (either disjunctive or conjunctive), a class of process templates $\templates$, e.g., templates of bounded size, and a class of properties, e.g. satisfaction of $h$-indexed \PromptmX formulas under a given fairness notion.
%, and a class of properties, which can be $h$-indexed formulas for some $h \in \Nat$ or the existence of (local/global) deadlocks. 

A \emph{cutoff} for a given class of systems with processes from $\templates$, a fairness notion $f \in \{lb,gb\}$ and a set of \PromptmX formulas $\Phi$ is a number $c \in \Nat$ such that

$$ \forall A,B \in \templates~ \forall \varphi \in \Phi~ \forall n \geq c :~ 
 \largesys \models_f \varphi ~\Leftrightarrow~ \cutoffsys \models_f \varphi.$$
%
%Similarly, a \emph{cutoff for (local/global) deadlock detection} for a class of systems with processes from $\templates$ is a number $c \in \Nat$ such that for all $A,B \in \templates$:
%$$ \largesys \textrm{ has a (local/global) deadlock } ~\Leftrightarrow~ \cutoffsys \textrm{ has a (local/global) deadlock.}$$

%Like the problem definitions above, cutoffs may additionally be flavoured with different notions of fairness.

Note that the existence of a cutoff implies that the PMCP is \emph{decidable} iff the model checking problem for the cutoff system $\cutoffsys$ is decidable.
Decidability of model checking for finite transition systems with specifications in \PromptmX and bounded fairness follows from the fact that bounded fairness can be expressed in \PromptmX, and from results on decidability of assume-guarantee model checking for \Prompt (cf. Kupferman et al.~\cite{KupfermanPV09} and Faymonville and Zimmermann~\cite{FaymonvilleZ17}[Lemmas 8, 9]).

\section{Cutoffs for Disjunctive Systems}
%
%Disjunctive systems are characterized by the interpretation of their guards. A transition of a process $p$ is enabled iff at least one local state of a  different process belongs to the guard of that transition.

In this section, we prove cutoff results for disjunctive systems under 
bounded fairness and stutter-insensitive specifications with or without 
promptness. To this end, in Section~\ref{sec:simulation-lemmas} we prove two 
lemmas that show how to simulate, up to bounded stuttering, 
local runs from a system of given 
size $n$ in a smaller or larger disjunctive system.
We then use these two lemmas in Subsections \ref{disj-LTL} and 
\ref{disj-prompt} to obtain cutoffs for specifications in \LTLmX and \PromptmX, 
respectively.

Moreover for the proofs of these two lemmas we utilize the same construction 
techniques that were used in~\cite{AJK16,AJK15,EmersonK00}, 
but in addition we analyze their effects on bounded fairness and bounded 
stutter equivalence. 
Note that we will only consider formulas of the form $\varphi(A,B^{(1)})$, 
however, as in previous work~\cite{AJK15,EmersonK00}, our results extend to 
specifications over an arbitrary number $h$ of $B$-processes.
%To this end, we show how to simulate (up to bounded stuttering) a given run in a system of size $n$ in a smaller or larger disjunctive system. 

Table~\ref{tab:disj} summarizes the results of this section: for specifications in $\LTLmX$ and $\PromptmX$ we obtain a cutoff that depends on the size of process template $B$, as well as on the number $h$ of quantified index variables. 
The table states 
\begin{wrapfigure}[12]{r}{0.3\linewidth}
%\vspace{-18pt}
\scalebox{0.95}{
%\begin{figure}[h]
\fboxrule=0pt
%\fboxsep=0pt
%\noindent

\fbox{%
\begin{minipage}[b]{0.5\linewidth}
%\centering
\begin{tikzpicture}[node distance=1.8cm,>=stealth,auto]
  \tikzstyle{state}=[circle,thick,draw=black,minimum size=7mm]
  \begin{scope}    
    \node [state] (r) {$\mathbf{r}$}
     edge [loop above] node[above = 0.5pt] {\tiny{$\exists \mathbf{nw}$}} (r);   
    \node [state] (nr) [below of = r] {$\mathbf{nr}$}
     edge [post, bend left] node[left = 0.5pt] {\tiny{$\exists \mathbf{nw}$}}  (r)
     edge [pre, bend right] (r); %   
  \end{scope}
  \end{tikzpicture}
  \caption{Reader} \label{fig:Reader}
\end{minipage}
}
\fbox{%
\begin{minipage}[b]{0.5\linewidth}
%\centering

  \begin{tikzpicture}[node distance=1.8cm,>=stealth,auto]
  \tikzstyle{state}=[circle,thick,draw=black,minimum size=7mm]

  \begin{scope}
    \node [state] (w) {$\mathbf{w}$}
     edge [loop above] node[above = 0.5pt] {\tiny{$\exists \mathbf{r}$}} (r);   
    \node [state] (nw) [below of = r] {$\mathbf{nw}$}
     edge [post, bend left]  (w)
     edge [pre, bend right] (w); %
   
  \end{scope}
  \end{tikzpicture}
  \caption{Writer} \label{fig:Writer}
  \end{minipage}
}
%\end{figure}
}
\label{fig:disj-reader-writer}
%\vspace{-28pt}
\end{wrapfigure} 
generalizations of Theorems~\ref{LTL-b-fair} and \ref{PromptLTL-b-fair} from 
the $2$-indexed 
case to the $h$-indexed case for arbitrary $h \in \Nat$. 
For one of the cases we were not able to obtain a cutoff result (as explained 
in Appendix~\ref{app:absence}.

\paragraph{\bf Simple Reader-Writer Example.}
Consider the disjunctive system $W {\parallel} R^n$, where $W$ is a writer 
process~(Figure~\ref{fig:Writer}), and $R$ is a reader 
process~(Figure~\ref{fig:Reader}). Let 
the specification $\varphi$ be
$\forall i \: \G(\mathbf{w} \rightarrow \F_{\p} [(\mathbf{w} \land \mathbf{nr}_i)] )$, i.e., 
if process $W$ is in state $\mathbf{w}$, then eventually all the $R$ 
processes will be in state $\mathbf{nr}$, while $W$ is in $\mathbf{w}$. 
According to Table~\ref{tab:disj}, the cutoff for checking whether $W {\parallel} R^n \models_{lb} \varphi$ is $5$.

\begin{table}[h]
\vspace{-1em}
\caption{Cutoffs for Disjunctive Systems}
\label{tab:disj}
\scriptsize
\centering
\begin{tabular}{lccc}
\toprule
& Local Bounded Fairness~~~ & Global Bounded Fairness   \\ 
\midrule
$h$-indexed \LTLmX & $2|Q_B| + h$ & $2|Q_B| + h$ \\
$h$-indexed \promptmX~~~ & $2|Q_B| + h$ & -\\
\bottomrule
\end{tabular}
\vspace{-1em}
\end{table}

\subsection{Simulation up to Bounded Stutter Equivalence}
\label{sec:simulation-lemmas}

\smartpar{Definitions.} Fix a run $x=x_0x_1...$ of the disjunctive system $\largesys$. Our constructions are based on the following definitions, where $q \in Q_B$:

\begin{itemize}
\item $\appears^{B_i}(q)$ is the set of all moments in $x$ where process $B_i$ is in state $q$: $\appears^{B_i}(q)=\{m \in \Nat \mid x_m(B_i)=q\}$.
\item $\appears(q)$ is the set of all moments in $x$ where at least one $B$-process is in state $q$: $\appears(q)=\{m \in \Nat \mid \exists i \in \{1,\ldots,n\}: x_m(B_i)=q\}$.%copy of $\mathcal{B}$
\item $f_q$ is the first moment in $x$ where $q$ appears: $f_q = min(\appears(q))$, and $\first_q \in \{1,\ldots,n\}$ is the index of a $B$-process where $q$ appears first, i.e., with $x_{f_q}(B_{\first_q}) = q$.
\item if $\appears(q)$ is finite, then $l_q = max(\appears(q))$ is the last moment where $q$ appears, and $\last_q \in \{1,\ldots,n\}$ is a process index with $x_{l_q}(B_{\last_q}) = q$ 
\item let $\visInf{} = \{q \in Q_B \mid \exists B_i \in \{B_2,\ldots,B_n\}: \appears^{B_i}(q) \textrm{ is infinite}\}$\\ and
$\visFin{} = \{q \in Q_B \mid \forall B_i \in \{B_2,\ldots,B_n\}: \appears^{B_i}(q) \textrm{ is finite}\}$.
\item $Set(x_i)$ is the set of all state that are visited by some process at moment $i$: $Set(x_i)=\{q|q \in (Q_A \cup Q_B)$ and $ \exists p \in \mP: x_i(p) = q \}$.
\end{itemize}

%%%%%%%%%%%%%%%%%%%%%%%%Monotonicity%%%%%%%%%%%%%Monotonicity%%%%%%%%%Monotonicity%%%%%%%%%%%%%%%%%%%%%%%%
Our first lemma states that any behavior of processes $A$ and $B_1$ in a system $\largesys$ can be simulated up to bounded stuttering in a system $\largesyse{n+1}$. This type of lemma is called a \emph{monotonicity lemma}.

\begin{lemma}[Monotonicity Lemma for Bounded Stutter Equivalence]\label{mon-disj}
Let $A, B$ be disjunctive process templates, $n \geq 2$, $b\in \Nat$ and $x \in \largesys$ with $\lbfair{b}(x,\{A,B_1\})$. Then there exists $y \in \largesyse{n+1}$ with $\lbfair{2b}(y,\{A,B_1\})$ and $x(A, B_1) \equiv_{2} y(A,B_1).$
\end{lemma}

\begin{proof}
Let $x$ be a run of $\largesys$ where $\lbfair{b}(x,\{A,B_1\})$. Let $y(A) = 
x(A)$ and $y(B_j) = x(B_j)$ for all $B_j \in \{B_1,\ldots,B_n\}$ and let the 
new process $B_{n+1}$ copy one of the $B$-processes of $\largesys$, i.e., 
$y(B_{n+1}) = x(B_i)$ for some $i \in \{1,\ldots,n\}$. Copying a local run violates the 
interleaving 
semantics as two processes will be moving at the same time. To solve this 
problem, we split every transition $(y_l, y_{l+1})$ where the interleaving 
semantics is violated by $B_i$ and $B_{n+1}$ executing local transitions 
$(q_i, g, q'_i)$ and $(q_{n+1}, g, q'_{n+1})$, respectively. To do this, 
replace $(y_l, y_{l+1})$ with two consecutive transitions $(y_l,u)(u,y_{l+1})$,
where $(y_l,u)$ is based on the local transition $(q_i, g, q'_i)$ 
and $(u,y_{l+1})$ is based on the local transition $(q_{n+1}, g, q'_{n+1})$. 
Note that both of these local transitions are enabled in the constructed run $y$
since the transition $(q_i, g, q'_i)$ is enabled in the original run $x$.
Moreover, run $y$ inherits unconditional fairness from $x$. 
Finally, it is easy to see that for 
every local transition of process $B_i$ in $x$, establishing interleaving 
semantics has added one additional stuttering step to every local run in $y$ 
including processes $A$ and $B_1$. Therefore we have that $\lbfair{2b}(y,\{A,B
_1\})$ and  $x(A, B_1) \equiv_{2} y(A,B_1)$.
\qed
\end{proof}

\paragraph{\bf Reader-Writer Example.}

\begin{wrapfigure}{r}{0.5\linewidth}
\vspace{-25pt}
\scalebox{0.95}{
\fboxrule=0pt
\fbox{%
\begin{minipage}[b]{0.48\linewidth}
\centering

  \begin{tikzpicture}
  \node at (0.3,1) {$t$};
  \draw (0.1,0.8) -- (0.45,0.8);
  \node at (0.3,0.5) {$0$};
  \node at (0.3,0.15) {$1$};
  \node at (0.3,-0.2) {$2$};
  \node at (0.3,-0.55) {$3$};
  \node at (0.3,-0.9) {$4$};
  \node at (0.3,-1.25) {$5$};
  \node at (0.3,-1.6) {$6$};
  %\node at (0.3,-1.95) {$7$};
  %\node at (0.3,-2.3) {\vdots};
  
  \draw (0.5,1.2) -- (0.5,-1.7);
  \draw[->] (0.2,-1.6) to [bend left = 45] (0.2,-0.55);
  
  \node at (1,1) {$W$};
  \draw (0.8,0.8) -- (1.15,0.8);
  \node at (1,0.5) {$\mathrm{nw}$};
  \node at (1,0.15) {$\mathrm{nw}$};
  \node at (1,-0.2) {$\mathrm{nw}$};
  \node at (1,-0.55) {$\mathbf{w}$};
  \node at (1,-0.9) {$\mathbf{nw}$};
  \node at (1,-1.25) {$\mathrm{nw}$};
  \node at (1,-1.6) {$\mathrm{nw}$};
  %\node at (1,-1.95) {$\mathbf{w}$};
  %\node at (1,-2.3) {\vdots};
  
  \node at (1.5,1) {$R_1$};
  \draw (1.3,0.8) -- (1.65,0.8);
  \node at (1.5,0.5) {$\mathrm{nr}$};
  \node at (1.5,0.15) {$\mathbf{r}$};
  \node at (1.5,-0.2) {$\mathrm{r}$};
  \node at (1.5,-0.55) {$\mathrm{r}$};
  \node at (1.5,-0.9) {$\mathrm{r}$};
  \node at (1.5,-1.25) {$\mathbf{r}$};
  \node at (1.5,-1.6) {$\mathrm{r}$};
  %\node at (1.5,-1.95) {$\mathrm{r}$};
  %\node at (1.5,-2.3) {\vdots};
  
  \node at (2,1) {$R_2$};
  \draw (1.8,0.8) -- (2.15,0.8);
  \node at (2,0.5) {$\mathrm{nr}$};
  \node at (2,0.15) {$\mathrm{nr}$};
  \node at (2,-0.2) {$\mathbf{r}$};
  \node at (2,-0.55) {$\mathrm{r}$};
  \node at (2,-0.9) {$\mathrm{r}$};
  \node at (2,-1.25) {$\mathrm{r}$};
  \node at (2,-1.6) {$\mathbf{r}$};
  %\node at (2,-1.95) {$\mathrm{r}$};
  %\node at (2,-2.3) {\vdots};
  \end{tikzpicture}
  \vspace{0.3cm}
  \caption{Run: $W || R^2$} \label{fig:MoRun2}
  \end{minipage}
}
\fbox{%
\begin{minipage}[b]{0.48\linewidth}
\centering

  \begin{tikzpicture}
  \node at (0.3,1) {$t$};
  \draw (0.1,0.8) -- (0.45,0.8);
  \node at (0.3,0.5) {$0$};
  \node at (0.3,0.15) {$1$};
  \node at (0.3,-0.2) {$2$};
  \node at (0.3,-0.55) {$3$};
  \node at (0.3,-0.9) {$4$};
  \node at (0.3,-1.25) {$5$};
  \node at (0.3,-1.6) {$6$};
  \node at (0.3,-1.95) {$7$};
  \node at (0.3,-2.3) {$8$};
  %\node at (0.3,-2.65) {\vdots};
  
  \draw (0.5,1.2) -- (0.5,-2.4);
  \draw[->] (0.2,-2.3) to [bend left = 45] (0.2,-0.9);
  
  \node at (1,1) {$W$};
  \draw (0.8,0.8) -- (1.15,0.8);
  \node at (1,0.5) {$\mathrm{nw}$};
  \node at (1,0.15) {$\mathrm{nw}$};
  \node at (1,-0.2) {$\mathrm{nw}$};
  \node at (1,-0.55) {$\mathrm{nw}$};
  \node at (1,-0.9) {$\mathbf{w}$};
  \node at (1,-1.25) {$\mathbf{nw}$};
  \node at (1,-1.55) {$\mathrm{nw}$};
  \node at (1,-1.95) {$\mathrm{nw}$};
  \node at (1,-2.3) {$\mathrm{nw}$};
  %\node at (1,-2.65) {\vdots};
  
  \node at (1.5,1) {$R_1$};
  \draw (1.3,0.8) -- (1.65,0.8);
  \node at (1.5,0.5) {$\mathrm{nr}$};
  \node at (1.5,0.15) {$\mathbf{r}$};
  \node at (1.5,-0.2) {$\mathrm{r}$};
  \node at (1.5,-0.55) {$\mathrm{r}$};
  \node at (1.5,-0.9) {$\mathrm{r}$};
  \node at (1.5,-1.25) {$\mathrm{r}$};
  \node at (1.5,-1.55) {$\mathbf{r}$};
  \node at (1.5,-1.95) {$\mathrm{r}$};
  \node at (1.5,-2.3) {$\mathrm{r}$};
  %\node at (1.5,-2.65) {\vdots};
  
  \node at (2,1) {$R_2$};
  \draw (1.8,0.8) -- (2.15,0.8);
  \node at (2,0.5) {$\mathrm{nr}$};
  \node at (2,0.15) {$\mathrm{nr}$};
  \node at (2,-0.2) {$\mathbf{r}$};
  \node at (2,-0.55) {$\mathrm{r}$};
  \node at (2,-0.9) {$\mathrm{r}$};
  \node at (2,-1.25) {$\mathrm{r}$};
  \node at (2,-1.55) {$\mathrm{r}$};
  \node at (2,-1.95) {$\mathbf{r}$};
  \node at (2,-2.3) {$\mathrm{r}$};
  %\node at (2,-2.65) {\vdots};
  
  \node at (2.5,1) {$R_3$};
  \draw (2.3,0.8) -- (2.65,0.8);
  \node at (2.5,0.5) {$\mathrm{nr}$};
  \node at (2.5,0.15) {$\mathrm{nr}$};
  \node at (2.5,-0.2) {$\mathrm{nr}$};
  \node at (2.5,-0.55) {$\mathbf{r}$};
  \node at (2.5,-0.9) {$\mathrm{r}$};
  \node at (2.5,-1.25) {$\mathrm{r}$};
  \node at (2.5,-1.55) {$\mathrm{r}$};
  \node at (2.5,-1.95) {$\mathrm{r}$};
  \node at (2.5,-2.3) {$\mathbf{r}$};
  %\node at (2.5,-2.65) {\vdots};
  \end{tikzpicture}
  \caption{Run: $W || R^3$} \label{fig:MoRun3}
  \end{minipage}
}
}
\vspace{-28pt}

\end{wrapfigure}

Consider the run $x$ of the system $W || R^2$ in Figure \ref{fig:MoRun2} where $W$ and $R$ are as defined in Figures~\ref{fig:Reader} and \ref{fig:Writer}.  We construct a run $y$ of the system $W || R^3$ (see Figure \ref{fig:MoRun3}) such that $x(W, R_1) \equiv_{2} y(W,R_1)$. The local run of process $R_3$ is obtained by (i) copying the run of $R_2$, and (ii) establishing the interleaving semantics as in the proof of Lemma \ref{mon-disj}.

\medskip

%which violates the specification $G(w \rightarrow F_{\p} [(w \land nr_1))$ as at time $t=4$ the process $W$ has left the writing state $w$ although the process $R_1$ has been in $r$ since the last occurence of the state $nw$

%%%%%%%%%%%%%%%%%%%%%%%%%%%%%%%%%%%%%%%%%%%%%%%%%%%%%%%%%%%%%%%%%%%%%%%%%%%%%%%

As mentioned in the above construction, if a local run of $x$ is $d$-bounded fair 
for some $d \in \Nat$, then it will be $2d$-bounded fair in the constructed 
run $y$. This observation leads to the following corollary.
%In the above construction one can notice that for every local transition of 
%process $B_i$ in $x$, establishing interleaving semantics has added one 
%additional stuttering step to every local run in $y$.
%This observation leads to the following corollary.
 
\begin{cor}\label{mon-disj-cor}
Let $A$, $B$ be disjunctive process templates, $n \geq 2$, $b \in \Nat$ and $x \in \largesys$ with $\bfair{b}(x)$.
Then there exists $y \in \largesyse{n+1}$ with $\bfair{2b}(y)$ and $x(A,B_1) \equiv_{2} y(A,B_1).$%(A,B_1,\ldots,B_n).$
\end{cor}

%%%%%%%%%%%%%%%%%%%%%%%%Bounding%%%%%%%%%%%%%Bounding%%%%%%%%%Bounding%%%%%%%%%%%%%%%%%%%%%%%%

Our second lemma is a \emph{bounding lemma} which states that any behavior of processes $A$ and $B_1$ in a disjunctive system $\largesys$ can be simulated up to bounded stuttering in a system $\largesyse{c}$, if $c$ is chosen to be sufficiently large and $n\geq c$.

\begin{lemma}[Bounding Lemma for Bounded Stutter Equivalence]\label{bound-disj}
Let $A, B$ be disjunctive process templates, $c = 2|Q_B| + 1$, $n \geq c$, $b \in \Nat$ and $x \in \largesys$ with $\lbfair{b}(x,\{A,B_1\})$. Then there exists $y \in \largesyse{c}$ with $\lbfair{(b \cdot c)}(y,\{A,B_1\})$ and $x(A, B_1) \equiv_{c} y(A,B_1)$.
\end{lemma}

\begin{proof}
Let $x$ be a run of $\largesys$ where $\lbfair{b}(x,\{A,B_1\})$. We show how 
to construct a run $y$ of $\largesyse{c}$ where $\lbfair{(b \cdot c)}
(y,\{A,B_1\})$ and $x(A, B_1) \equiv_{c} y(A,B_1)$.

The basic idea is that, in order to ensure that all transitions in $y$
are enabled at the time they are taken, we ``flood'' every 
state $q$ that is visited in $x$ with one or more processes that 
enter $q$ and stay there. Additionally, we need to take care of 
fairness, which requires a more complicated construction that allows every 
such process to move infinitely often. Therefore, some processes have to 
leave the state they have flooded (if that state only appears finitely often 
in the original run), and every process needs to eventually enter a loop that 
allows it to move infinitely often. In the following, we construct such runs 
formally.

\textbf{Construction:}
\begin{enumerate}
\item \textbf{(Flooding with evacuation)}: To every $q \in \visFin{}(x)$, devote one process $B_{i_q}$ that copies $B_{\first_q}$ until the time $f_q$, then stutters in $q$ until time $l_q$ where it starts copying $B_{last_q}$ forever. Formally: 
 $$y(B_{i_q})=x(B_{\first_q})[0:f_q].(q)^{l_q-f_q}.x(B_{last_q})[l_q+1:\infty]$$
%Note that the number of consecutive stuttering steps introduced here is bounded by $|u|$.
%
\item \textbf{(Flooding with fair extension)}: For every $q \in \visInf{}(x)$, let $B^{inf}_q$ be a process that visits $q$ infinitely often in $x$. We devote to $q$ two processes $B_{i_{q_1}}$ and $B_{i_{q_2}}$ that both copy $B_{\first_q}$ until the time $f_q$, and then stutter in $q$ until $B^{inf}_q$ reaches $q$ for the first time. After that, let $B_{i_{q_1}}$ and $B_{i_{q_2}}$ copy $B^{inf}_q$ in turns as follows: $B_{i_{q_1}}$ copies $B^{inf}_q$ until it reaches $q$ while $B_{i_{q_2}}$ stutters in $q$, then $B_{i_{q_2}}$ copies $B^{inf}_q$ until it reaches $q$ while $B_{i_{q_1}}$ stutters in $q$ and so on. %Note that the number of consecutive stuttering steps introduced here for a given process is bounded by $|u|+2|v|$ because $B^{inf}_q$ needs up to $|u|+|v|$ steps to reach $q$, and one of the processes has to wait for up to $|v|$ additional global steps before it can move. 

\item Establish interleaving semantics as in the proof of Lemma \ref{mon-disj}.
\end{enumerate}

After steps 1 and 2, the following property holds: at any time $t$ we have 
that $Set(x_t) \subseteq Set(y_t)$, which guarantees that every transition 
along the run is enabled. Note that establishing the interleaving semantics 
preserves this property.

Finally, establishing interleaving semantics could introduce additional 
stuttering steps to the local runs of processes $A$ and $B_1$ whenever steps 1
 or 2 of the construction use the same local run from $x$ more than once 
(e.g. if $\exists q_i,q_j \in Q_B$ with $\first_{q_i} = \first_{q_j}$). A local 
run of $x$ can be used in the above construction at most $2|Q_B|$ times, 
therefore we have $x(A, B_1) \equiv_{c} y(A,B_1)$. Moreover, since the upper 
bound of consecutive stuttering steps in $A$ or $B_1$ is 
$(2|Q_B| + 1) \cdot b$, we get $\lbfair{(b \cdot c)}(y,\{A,B_1\})$.
\qed
\end{proof}
\paragraph{\bf Reader-Writer Example.}
Consider again the reader-writer system in Figures~\ref{fig:Reader} and 
\ref{fig:Writer}. For any run $x$ of $W{\parallel}R^n$, using the construction above we obtain a 
run $y$ of $W{\parallel}R^5$ (or even a smaller system) with $x(W, R_1) \equiv_{5} y(W,R_1)$.  

\subsection{Cutoffs for Specifications in \LTLmX under Bounded Fairness}\label{disj-LTL}

The PMCP for disjunctive systems with specifications from \LTLmX has 
been considered in several previous 
works~\cite{AJK16,EmersonK00,JS18}.
In the following we extend these results by proving cutoff results under 
bounded fairness.

\begin{theorem}[Cutoff for \LTLmX with Global Bounded Fairness]\label{LTL-b-fair}
Let $A$, $B$ be disjunctive process templates, $c = 2|Q_B| + 1$, $n \geq c$, and $\varphi(A, B^{(1)})$ a specification with $\varphi \in$ \LTLmX. Then:
$$ \left( \forall b \in \Nat :~ \fairSys{b}{n} \models \varphi(A, B^{(1)}) \right) \iff \left( \forall b' \in \Nat :~ \fairSys{b'}{c} \models \varphi(A, B^{(1)}) \right)$$
%$\forall b \in \Nat~ \forall x \in \largesys:  \bfair{b}(x) \implies x \models \varphi(A, B^{(1)})$\\ $\iff $ \\ $\forall b' \in \Nat~ \forall y \in \largesyse{2|Q_B|}:  \bfair{b'}(y) \implies y \models \varphi(A, B^{(1)}).$
\end{theorem}

\noindent
We prove the theorem by proving two lemmas, one for each direction of the equivalence. %Note that $\largesys$ $\nbmodels$ $\varphi(A, B^{(1)})$ iff $\exists b \in \Nat~ \exists x \in \largesys:~ \bfair(x) \land x \not \models \varphi(A, B^{(1)}).$

%%%%%%%%%%%%%%%%%%%%%%%%Monotonicity%%%%%%%%%%%%%Monotonicity%%%%%%%%%Monotonicity%%%%%%%%%%%%%%%%%%%%%%%%
\begin{lemma}[Monotonicity Lemma for \LTLmX]\label{mon-disj-LTL-b-fair}
Let $A$, $B$ be disjunctive process templates, $n\geq 1$, and $\varphi(A, B^{(1)})$ a specification with $\varphi \in $ \LTLmX. Then:
$$\left(\exists b \in \Nat :~ \fairSys{b}{n} \not \models \varphi(A, B^{(1)}) \right) \Rightarrow \left( \exists b' \in \Nat :~ \fairSys{b'}{n+1} \not \models \varphi(A, B^{(1)})\right)$$
%$$ \forall n \geq 1: \left( \largesys \nbmodels \varphi(A, B^{(1)}) \Rightarrow \largesyse{n+1} \nbmodels \varphi(A, B^{(1)}) \right).$$
\end{lemma}
\begin{proof}

Assume $\exists b \in \Nat :~ \fairSys{b}{n} \not \models \varphi(A, B^{(1)})$. 
Then there exists a run $x$ of $\largesys$ where $x$ is $\bfair{b}(x)$ and $x 
\not \models \varphi(A, B^{(1)})$. According to Corollary \ref{mon-disj-cor} 
there exists $y$ of $\largesyse{n+1}$ where $\bfair{2b}(y)$ and  $x(A, B_1) 
\equiv_{2} y(A,B_1)$, which guarantees that $y \not \models \varphi(A, B^{(1)}
)$.\qed
\end{proof}

%%%%%%%%%%%%%%%%%%%%%%%%Bounding%%%%%%%%%%%%%Bounding%%%%%%%%%Bounding%%%%%%%%%%%%%%%%%%%%%%%%

For the corresponding bounding lemma, our construction is based on that of 
Lemma~\ref{bound-disj}. However, the local runs resulting from that 
construction might stutter in some local states for an unbounded time (e.g. 
local runs devoted for states in $\visFin{F}$). To bound stuttering in such 
constructions, given an arbitrary run of a system $\largesys$, we first show 
that whenever there exists a bounded-fair run that violates a specification 
in \LTLmX, then there also exists an ultimately periodic run with the same 
property.

%\begin{definition}[non-deterministic B\"uchi automaton]
%\sj{shorten this def. if we need space}
A (non-deterministic) \emph{B\"uchi automaton} is a tuple $\mA = (\Sigma ,Q_{\mA}, \delta, 
a_0, \alpha)$, where $\Sigma$ is a finite alphabet, $Q_{\mA}$ is a finite set of 
states, $\delta: Q_{\mA} \times \Sigma \rightarrow 2^{Q_{\mA}}$ is a transition function, 
$a_0 \in Q_{\mA}$ is an initial state, and $\alpha \subseteq Q_{\mA}$ is a B\"uchi 
acceptance condition. Given an \LTL specification $\varphi$, we denote by $\buchi{\varphi}$ the B\"uchi automaton that accepts exactly all words that satisfy $\varphi$ \cite{vardi1986automata}.
%\end{definition}

\begin{lemma}[Ultimately Periodic Counter-Example]\label{periodic-CE}
Let $\varphi \in \LTL$ and $b \in \Nat$. If 
$\fairSys{b}{n} \not \models \varphi$ then there exists a run $x=uv^{\omega}$ 
of $\largesys$ with $\bfair{b}(x)$, and $x \not \models \varphi$, where $u,v$ 
are finite paths, and $|u|, |v| \leq 2 \cdot |Q_A| \cdot |Q_B|^n \cdot b^{n+1}
 \cdot |Q_{\buchi{\neg \varphi}}|$.
\end{lemma}

A formal proof of the lemma can be found in Appendix~\ref{app:proof4}.
Now, we have all the ingredients to prove the bounding lemma for the case of 
\LTLmX specifications and (global) bounded fairness.

\begin{lemma}[Bounding Lemma for \LTLmX]\label{bound-disj-LTL-b-fair-1}
Let $A$, $B$ be disjunctive process templates, $c = 2|Q_B| + 1$, $n \geq c$, and $\varphi(A, B^{(1)})$ a specification with $\varphi \in$ \LTLmX. Then:
$$\left(\exists b \in \Nat :~ \fairSys{b}{n} \not \models \varphi(A, B^{(1)}) \right) \Rightarrow \left(\exists b' \in \Nat :~ \fairSys{b'}{c} \not \models \varphi(A, B^{(1)})\right)$$
\end{lemma}
\begin{proof}
Assume $\exists b \in \Nat: \fairSys{b}{n} \not \models \varphi(A, B^{(1)})$. Then by Lemma~\ref{periodic-CE} there is a run $x=uv^\omega$ of $\largesys$, where $\bfair{b}(x)$ and $|u|, |v| \leq 2 \cdot |Q_A| \cdot |Q_B|^n \cdot b^{n+1} \cdot |Q_{\buchi{\neg \varphi}}|$. According to Lemma \ref{bound-disj}, we can construct out of $x$ a run $y$ of $\largesyse{c}$ where $\lbfair{b''}(y,\{A,B_1\})$, and $x(A, B_1) \equiv_{d} y(A,B_1)$ with $d = 2|Q_B| + 1$ and $b''= b \cdot d$. The latter guarantees that $y \not \models \varphi(A, B^{(1)})$.  We still need to show that $\bfair{b'}(y)$ for some $b' \in \Nat$. As $x=uv^\omega$, we observe that the construction of Lemma~\ref{bound-disj} ensures the following:
\begin{itemize}
\item The number of consecutive stuttering steps per process introduced in step 1 is bounded by $|u|$.
\item The number of consecutive stuttering steps introduced in step 2 for a given process is bounded by $|u|+2|v|$ because $B^{inf}_q$ needs up to $|u|+|v|$ steps to reach $q$, and one of the processes has to wait for up to $|v|$ additional global steps before it can move. 
\end{itemize} 

In addition to the stuttering steps introduced in step $1$ and $2$, if more than one of the constructed processes simulate the same local run of $x$ then establishing the interleaving semantics would be required, which in turn introduces additional stuttering steps. Therefore the upper bound of consecutive stuttering steps introduced in step $3$ of the construction is $(2|Q_B| + 1)  \cdot b$. Therefore $\bfair{b'}(y)$ where $b' = (2|Q_B| + 1) \cdot b + 6 \cdot |Q_A| \cdot |Q_B|^n \cdot b^{n+1} \cdot |Q_{\buchi{\neg \varphi}}|$.
\qed
\end{proof}

\begin{remark}
With a more complex construction that uses a stutter-insensitive automaton 
$\mathcal{A}$~\cite{etessami99} to represent the specification and considers 
runs of the composition of system and automaton, we can obtain a much smaller 
$b'$ that is also independent of $n$. This is based on the observation that if 
in $y$ some process is consecutively stuttering for more than $|\largesyse{c} 
\times \mathcal{A}|$ steps, then there must be a repetition of states from 
the product in this time, and we can simply cut the infix between the 
repeating states from the constructed run $y$.
\end{remark}

\subsection{Cutoffs for Specifications in \promptmX}\label{disj-prompt}

\LTL specifications cannot enforce boundedness of the time that elapses 
before a liveness 
property is satisfied. Prompt-LTL solves this problem by introducing the 
prompt-eventually operator explained in Section~\ref{promptLTL definition}. 
Since we consider concurrent asynchronous systems, the satisfaction of a 
Prompt-LTL formula can also depend on the scheduling of processes. If 
scheduling can introduce unbounded delays for a process, then promptness can 
in general not be guaranteed. Hence, non-trivial \Prompt specifications can 
\emph{only} be satisfied under the assumption of bounded fairness, and therefore this 
is the only case we consider here.

\begin{theorem}[Cutoff for \PromptmX with Local Bounded Fairness]\label{PromptLTL-b-fair}
Let $A$, $B$ be disjunctive process templates, $c=2|Q_B| + 1$ $n \geq c$, and $\varphi(A, B^{(1)})$ a specification with $\varphi \in$ \PromptmX. Then:
$$\largesyse{c} \lbmodels \varphi(A, B^{(1)}) ~\iff~ \largesys \lbmodels \varphi(A, B^{(1)}).$$
\end{theorem}

\noindent
Again, we prove the theorem by proving a monotonicity and a bounding lemma. Note that $\largesys$ $\nlbmodels$ $\varphi(A, B^{(1)})$ iff 
$$\exists b \in \Nat~ \forall k \in \Nat~ \exists x \in \largesys{:}~ \lbfair{b}(x,\{A, B^{(1)}\}) \land (x,0,k) \not \models \varphi(A, B^{(1)}).$$

%%%%%%%%%%%%%%%%%%%%%%%%Monotonicity%%%%%%%%%%%%%Monotonicity%%%%%%%%%Monotonicity%%%%%%%%%%%%%%%%%%%%%%%%
\begin{lemma}[Monotonicity Lemma for \PromptmX]\label{mon-disj-PromptLTL}
Let $A$, $B$ be disjunctive process templates, $n \geq 2$, and $\varphi(A, B^{(1)})$ a specification with $\varphi \in $ \PromptmX. Then:
%$$ \forall n \geq 1: \left( \largesys \ngbmodels \varphi(A, B^{(1)}) \Rightarrow \largesyse{n+1} \ngbmodels \varphi(A, B^{(1)}) \right).$$
$$ \largesys \nlbmodels \varphi(A, B^{(1)}) ~\Rightarrow~ \largesyse{n+1} \nlbmodels \varphi(A, B^{(1)}) .$$
\end{lemma}
\begin{proof}

Assume $\largesys \nlbmodels \varphi(A, B^{(1)})$. Then there exists $b \in 
\Nat$ such that $\forall k \in \Nat$ there is a run $x$ of $\largesys$ where 
$\lbfair{b}(x,\{A, B^{(1)}\})$, and $(x,0,2 \cdot k) \not \models 
\varphi(A, B^{(1)})$. Then according to Lemma \ref{mon-disj} there exists $y$ of $\largesyse{n+1}$ where $\lbfair{2b}(y,\{A, B^{(1)}\})$ and  $x(A, B_1) \equiv_{2} y(A,B_1)$, which guarantees, according to Corollary \ref{cor-stutter}, that $(y,0,k) \not \models \varphi(A, B^{(1)})$. As a consequence there exists $b \in \Nat$ such that $\forall k \in \Nat$ there is a run $y$ of $\largesyse{c}$ where $\lbfair{2b}(y,\{A, B^{(1)}\})$ and $(y,0, k) \not \models 
\varphi(A, B^{(1)})$, thus $\largesyse{c} \nlbmodels \varphi(A, B^{(1)})$.
\qed
\end{proof}

Using the same argument of the above proof but by using Corollary \ref{mon-disj-cor} instead of Lemma \ref{mon-disj} to construct the globally bounded fair counter example, we obtain the following: 
%The above lemma and Corollary~\ref{mon-disj-cor} 
%\sj{how can Lemma 7, which uses lfair, be combined with Corollary 2, which uses gfair?}
\begin{cor}\label{mon-disj-PromptLTL-cor}
Let $A$, $B$ be disjunctive process templates, $n \geq 2$, and $\varphi(A, B^{(1)})$ a specification with $\varphi \in $ \PromptmX. Then:
$$\largesys \ngbmodels \varphi(A, B^{(1)}) ~\Rightarrow~ \largesyse{n+1} \ngbmodels \varphi(A, B^{(1)}) .$$
\end{cor}

%%%%%%%%%%%%%%%%%%%%%%%%Bounding%%%%%%%%%%%%%Bounding%%%%%%%%%Bounding%%%%%%%%%%%%%%%%%%%%%%%%

\begin{lemma}[Bounding Lemma for \PromptmX]\label{bound-disj-ProomptLTL-lb}
Let $A$, $B$ be disjunctive process templates, $c=2|Q_B| + 1$,  $n \geq c$, and $\varphi(A, B^{(1)})$ a specification with $\varphi \in \PromptmX$. Then:
$$ \largesys \nlbmodels \varphi(A, B^{(1)}) ~\Rightarrow~ \largesyse{c} \nlbmodels \varphi(A, B^{(1)}) .$$
\end{lemma}
\begin{proof}

Assume $\largesys \nlbmodels \varphi(A, B^{(1)})$. Then there exists $b \in 
\Nat$ such that $\forall k \in \Nat$ there is a run $x$ of $\largesys$ where 
$\lbfair{b}(x,\{A, B^{(1)}\})$ and $(x,0,d \cdot k) \not \models 
\varphi(A, B^{(1)})$ with $d =(2|Q_B| + 1)$. According to Lemma \ref{bound-disj} we can construct for every such $x$ a run $y$ of $\largesyse{c}$ where $\lbfair{(d\cdot b)}(y,\{A, B^{(1)}\})$,  and $x(A, B_1) \equiv_{d} y(A,B_1)$, which guarantees that $(y, 0, k) \not \models \varphi(A, B^{(1)})$ (see Corollary \ref{cor-stutter}). Thus, there exists $b \in 
\Nat$ such that $\forall k \in \Nat$ there is a run $y$ of $\largesyse{c}$ where $\lbfair{(d\cdot b)}(y,\{A, B^{(1)}\})$ and $(y,0, k) \not \models 
\varphi(A, B^{(1)})$, thus $\largesyse{c} \nlbmodels \varphi(A, B^{(1)})$.
\qed
\end{proof}

\section{Cutoffs for Conjunctive Systems}
%Conjunctive systems are distinguished by the interpretation of their guards. A transition of a process $p$ is enabled if all local states of the other processes belong to the guard of that transition.

In this section we investigate cutoff results for conjunctive systems under 
bounded fairness and specifications in \PromptmX.
Table~\ref{tab:conj} summarizes the results of this section, as generalizations of Theorems~\ref{Theorem-conj-PROMPT-LTL-lb-fair} and \ref{Theorem-conj-PROMPT-LTL-gb-fair} to $h$-indexed specifications. Note that for results marked with a $*$ we require processes to be \emph{bounded initializing}, i.e., that every cycle in the process template contains the initial state.\footnote{This is only slightly more restrictive than the assumption that they are initializing, as stated in the definition of conjunctive systems in Section~\ref{sec:model}.}

\begin{table}[ht]
\vspace{-1em}
\caption{Cutoffs for Conjunctive Systems}
\label{tab:conj}
\scriptsize
\centering
\begin{tabular}{lccc}
\toprule
& Local Bounded Fairness~~~ & Global Bounded Fairness   \\ 
\midrule
$h$-indexed \LTLmX & $h+1$ & $h+1^*$ \\
$h$-indexed \promptmX~~~ &$h+1$& $h+1^*$ \\
\bottomrule
\end{tabular}
\vspace{-1em}
\end{table}

\subsection{Cutoffs under Local Bounded Fairness}

\begin{theorem}[Cutoff for \PromptmX with Local Bounded Fairness]\label{Theorem-conj-PROMPT-LTL-lb-fair}
Let $A, B$ be conjunctive process templates, $n \geq 2$, and $\varphi(A, B^{(1)})$ a 
specification with $\varphi \in \PromptmX$. Then:
$$\largesyse{2} \lbmodels \varphi(A, B^{(1)}) ~\iff~ \largesys \lbmodels 
\varphi(A, B^{(1)}).$$
\end{theorem}

We prove the theorem by proving two lemmas, one for each direction of the equivalence. 
Note that $\largesys$ $\nlbmodels$ $\varphi(A, B^{(1)})$ iff $\exists b \in \Nat~ \forall k \in \Nat~ \exists x \in \largesys:~ \bfair{b}(x) \land (x,0,k) \not \models \varphi(A, B^{(1)}).$

\begin{lemma}[Monotonicity Lemma, \PromptmX with Local Bounded Fairness]\label{mon-conj-PROMPT-LTL-lb-fair}
Let $A, B$ be conjunctive process templates, $n \geq 2$, and $\varphi(A, B^{(1)})$ a 
specification with $\varphi \in \PromptmX$. Then:
$$ \largesys \nlbmodels \varphi(A, B^{(1)}) ~\Rightarrow~ \largesyse{n+1} \nlbmodels \varphi(A, B^{(1)}).$$
\end{lemma}

\begin{proof}
Assume $\largesys$ $\nlbmodels$ $\varphi(A, B^{(1)})$. Then there exists $b 
\in \Nat$ such that $\forall k \in \Nat$ there is a run $x$ of $\largesys$ 
where $\bfair{b}(x)$ and $(x,0, k) \not \models \varphi(A, B^{(1)})$. For 
every such $x$, we construct a run $y$ of $\largesyse{n+1}$ with $\lbfair
{b}(y)$ and $(y,0,k) \not \models \varphi(A, B^{(1)})$.
Let $y(A) = x(A)$ and $y(B_j) = x(B_j)$ for all $B_j \in \{B_1,\ldots,B_n\}$ 
and let the new process $B_{n+1}$ "share" a local run $x(B_i)$ with an 
existing process $B_i$ of $\largesyse{n+1}$ in the following way: one process stutters in $init_B$ 
while the other makes transitions from $x(B_i)$, and whenever $x(B_i)$ enters 
$init_B$ the roles are reversed. Since this changes the behavior of $B_i$, 
$B_i$ cannot be a process that is mentioned in the formula, i.e. we need $n \geq 2$ for a 
formula $\varphi(A,B^{(1)})$. Then we have $\lbfair{b}(y,\{A,B_1\})$ as the 
run of $B_{n+1}$ inherits the unconditional fairness behavior from the local 
run of the process $B_i$ in $x$. Note that it is not guaranteed that the 
local runs $y(B_i)$ and $y(B_{n+1})$ are bounded fair as the time between two 
occurrences of $init_{B}$ in $x(B_i)$ is not bounded. Moreover we have 
$x(A,B_1) \equiv_{1} y(A, B_1)$, which according to Corollary~\ref{cor-stutter} implies
$(y(A, B_1), k) \not \models  \varphi(A, B^{(1)})$. 
\qed
\end{proof}  

%For the following results, we need one additional observation.

%\begin{lemma}[Fairness and Stuttering]\label{fair-stutter}\ms{We ar using this in the proof of lemma 10 however there we have that b = b', therefore I think we can just remove it as it is very obvious, it was also used in lemma 11 but I removed it as the value of d was wrong}
%Given two system runs $x \in \largesys$ and $y\in \largesyse{m}$, if $\bfair{b}(x)$, $\bfair{b'}(y)$, and $\exists p_i,p_j \in \mP$ s.t. $x(p_i)$ and $y(p_j)$ are stutter equivalent then $x(p_i) \equiv_{max\{b,b'\}} y(p_j)$.
%\end{lemma}

%\begin{proof}
%The runs $x$ and $y$ are bounded fair and stutter equivalent then we have: $x=x_0x_1 \ldots$,
%$y=y_0y_1 \ldots$ and $\forall i \in \Nat~ \exists \alpha \in \AP$ where $x_i = \alpha^e$, $y_i = \alpha^f$ for some $e,f \in \Nat$. Moreover we have $\forall i \in \Nat~ |x_i| < b$ and $|y_i| < b'$ then:
%\begin{itemize}
%\item if $b' \leq b$ then  $\forall i \in \Nat~ |x_i| \leq b \cdot |y_i|$ and $|y_i| \leq b \cdot |x_i|$
%\item if $b \leq b'$ then  $\forall i \in \Nat~ |x_i| \leq b' \cdot |y_i|$ and $|y_i| \leq b' \cdot |x_i|$
%\end{itemize} 
%\end{proof}

\begin{lemma}[Bounding Lemma, $\PromptmX$, Local Bounded Fairness]\label{bound-conj-PROMPT-LTL-lb-fair}
Let $A, B$ be conjunctive process templates, $n \geq 1$, and $\varphi(A, B^{(1)})$ a 
specification with $\varphi \in \PromptmX$. Then:
$$\largesys \nlbmodels \varphi(A, B^{(1)}) ~\Rightarrow~ \largesyse{1} \nlbmodels \varphi(A, B^{(1)}).$$
\end{lemma}
\begin{proof}
Assume $\largesys$ $\nlbmodels$ $\varphi(A, B^{(1)})$. Then there exists $b 
\in \Nat$ such that $\forall k \in \Nat$ there is a run $x$ of $\largesys$ 
where $\bfair{b}(x)$, and $(x,0,b \cdot k) \not \models \varphi(A, B^{(1)})$
. For every such $x$, we construct a run $y$ in the cutoff system $
\largesyse{1}$ by copying the local runs of processes $A$ and $B_1$ in $x$ 
and deleting stuttering steps. It is easy to see that $\bfair{b}(y)$ then we 
have $x(A,B_1) \equiv_b y(A, B_1)$, and by Corollary \ref{cor-stutter} 
$(y(A, B_1), k) \not \models  \varphi(A, B^{(1)})$. \qed
\end{proof}

Note that this is the same proof construction as in Au{\ss}erlechner et al.~\cite{AJK16}, and we simply observe that this construction preserves bounded fairness.

\subsection{Cutoffs under Global Bounded Fairness}

As mentioned before, to obtain a result that preserves global bounded fairness, we need to restrict process template $B$ to be bounded initializing.

\begin{theorem}[Cutoff for \PromptmX with Global Bounded Fairness]\label{Theorem-conj-PROMPT-LTL-gb-fair}
Let $A, B$ be conjunctive process templates, where $B$ is bounded initializing, $n \geq 2$, and $\varphi(A, B^{(1)})$ a specification with $\varphi \in \PromptmX$. Then:
$$\largesyse{2} \gbmodels \varphi(A, B^{(1)}) ~\iff~ \largesys \gbmodels \varphi(A, B^{(1)}).$$
\end{theorem}

Again, the theorem can be separated into two lemmas.

\begin{lemma}[Monotonicity Lemma, \PromptmX, Global Bounded Fairness]
\label{mon-conj-PROMPT-LTL-gb-fair}
Let $A, B$ be conjunctive process templates, where $B$ is bounded initializing, $n \geq 2$, and $\varphi(A, B^{(1)})$ a 
specification with $\varphi \in \PromptmX$. Then:
$$\largesys \ngbmodels \varphi(A, B^{(1)}) ~\Rightarrow~ \largesyse{n+1} 
\ngbmodels \varphi(A, B^{(1)}).$$
\end{lemma}
\begin{proof}
Assume $\largesys$ $\ngbmodels$ $\varphi(A, B^{(1)})$. Then there exists $b 
\in \Nat$ such that $\forall k \in \Nat$ there is a run $x$ of $\largesys$ 
where $\bfair{b}(x)$, and $(x,0, (b + |Q_B|) \cdot k) \not \models \varphi(A, 
B^{(1)})$. For every such $x$, we construct a run $y$ of $\largesyse{n+1}$ in 
the same way we did in the proof of Lemma \ref{mon-conj-PROMPT-LTL-lb-fair}. 
Then we have $\bfair{b'}(y)$ with $b' = b + |Q_B|$ as $init_B$ is on every 
cycle of the process template $B$. Moreover we have 
$x(A,B_1) \equiv_{1} y(A, B_1)$ which according to 
Corollary~\ref{cor-stutter} implies that $(y(A, B_1), k) \not \models  
\varphi(A, B^{(1)})$. 
\qed
\end{proof}

\begin{lemma}[Bounding Lemma, \PromptmX, Global Bounded Fairness]
\label{bound-conj-PROMPT-LTL-gb-fair}
Let $A, B$ be conjunctive process templates, where $B$ is bounded initializing, $n \geq 1$, and $\varphi(A, B^{(1)})$ a 
specification with $\varphi \in \PromptmX$. Then:
$$\largesys \ngbmodels \varphi(A, B^{(1)}) ~\Rightarrow~ \largesyse{1} \ngbmodels \varphi(A, B^{(1)}).$$
\begin{proof}
Under the given assumptions, we can observe that the construction from 
Lemma~\ref{bound-conj-PROMPT-LTL-lb-fair} also preserves global bounded fairness.
\end{proof}
\end{lemma}

%\begin{cor}[Fairness Cutoff and Promptness bound]
%Given two process templates $A$ and $B$ and a specification $\varphi(A, B^{(1)}) \in$ \PromptmX where $\largesyse{2}$ $\gbmodels$ $\varphi(A, B^{(1)})$ then given $b,k \in \Nat~$:
%\begin{center}
%\textbf{if} $\forall x \in \largesyse{1}:~ \bfair(x) \Rightarrow (x,0,k) \models \varphi(A, B^{(1)}).$ \\
%\textbf{then} $\forall n \geq 1$ $\forall y \in \largesys:~ \bfair(y) \Rightarrow (y,0,b \cdot k) \models \varphi(A, B^{(1)}).$

%\end{center}
%\end{cor}

%Since LTL$\setminus$X is a fragment of \promptmX, Theorem \ref{Theorem-conj-PROMPT-LTL-gb-fair} awards the following corollary for \LTLmX specifications under bounded fairness:\sj{I'm not sure if we need the corollary, it is obvious}
%
%\begin{cor}[Cutoff for \LTLmX with Global Bounded Fairness]
%Let $A, B$ be process templates, $n \geq 2$, and $\varphi(A, B^{(1)})$ a 
%specification with $\varphi \in \LTLmX$. Then:
%$$\largesyse{2} \models \varphi(A, B^{(1)}) ~\iff~ \largesys \models \varphi(A, B^{(1)}).$$
%\sj{this was only an implication before, but why have a corollary for the bounding lemma?}
%\ms{with which restriction?}
%\end{cor}

%The lemma motivates the following model checking strategy: Given a system $\largesys$, a fairness bound $b$ and a PROMPT-LTL$\setminus$X formula $\varphi(A, B^{(1)})$, it would be sufficient to model check the system $\largesyse{2}$ with bound $b' = b - n + 2$, and if the model checker returned true and returned some eventuality bound $k'$ then we can deduce that:
%$$\forall x \in \largesys:~ \bfair(x) \Rightarrow (x,0,b.k') \models \varphi.$$

\section{Token Passing Systems}
\label{sec:TPS}

In this section, we first introduce a system model for token passing systems and then show how to obtain cutoff results for this class of systems.

\subsection{System Model}

\smartpar{Processes.} A \emph{token passing process} is a transition system $T=(\qt, \initt, \sigt, \delta )$ where
\begin{itemize}
	\item $\qt = \overline{\qt} \times \{0,1\}$ is a finite set of states. $\overline{\qt}$ is a finite non-empty set. The boolean component $\{0,1\}$ indicates the possession of the token.
	\item $\initt$ is the set of initial states with $\initt \cap (\overline{\qt} \times \{0\}) \neq \emptyset$ and $\initt \cap (\overline{\qt} \times \{1\}) \neq \emptyset$.
	\item $\Sigma_T = \{\epsilon, rcv, snd\}$ is the set of actions, where $\epsilon$ is an asynchronous action, and $\{rcv, snd\}$ are the actions to receive and send the token.
	\item $\deltat{T} = \qt \times \Sigma_T \times \qt$ is a transition relation, such that $((q,b), a, (q',b')) \in \deltat{T}$ iff all of the following hold:
	\begin{itemize}
	\item $a = \epsilon ~\Rightarrow~ b = b'$.
	\item $a=snd ~\Rightarrow~ b = 1$ and $b' = 0$
	\item $a=rcv ~\Rightarrow~ b = 0$ and $b' = 1$
	\end{itemize}
	
\end{itemize}
 
\smartpar{Token Passing System.} Let $G = (V,E)$ be a finite directed graph without self loops where $V=\{1,\ldots, n\}$ is the set of vertices, and $E \subseteq V \times V$ is the set of edges. A \emph{token passing system} $T^{n}_G$ is a concurrent system containing $n$ instances of process $T$ where the only synchronization between the processes is the sending/receiving of a token according to the graph $G$. Formally, $T^{n}_G = (S,init_S, \Delta)$ with:
\begin{itemize}
\item $S = (\qt)^n$.
\item $init_S = \{ s \in (\initt)^n$ such that exactly one process holds the token $\}$,
%\item $init_S \subseteq (\initt)^n$ such that exactly one process holds the token at any time.
\item $\Delta \subseteq S \times S$ such that $((q_1,\ldots,q_n),(q'_1,\ldots, q'_n)) \in \Delta$ iff:
\begin{itemize}
\item \textbf{Asynchronous Transition}. $\exists i \in V$ such that $(q_i,\epsilon,q'_i) \in \deltat{T_i}$, and $\forall j \neq i$ we have $q_j = q'_j$.
\item \textbf{Synchronous Transition}. $\exists (i,j) \in E$ such that $(q_i,snd,q'_i) \in \deltat{T_i}$, $(q_j,rcv,q'_j) \in \deltat{T_j}$, and $\forall z \in V \setminus \{i,j\}$ we have $q_z = q'_z$.
\end{itemize}
\end{itemize}

\smartpar{Runs.} A \emph{configuration} of a system $T^{n}_G$ is a tuple $(s,ac)$ where $s \in S$, and either $ac = a_i$ with $a \in \sigt$, and $i \in V$ is a process index, or $ac = (snd_i, rcv_j)$ where $i,j \in V$ are two process indices with $i \neq j$. A run is an infinite sequence of configurations $x = (s_0,ac_0)(s_1,ac_1) \ldots$ where $s_0 \in init_S$ and $s_{i+1}$ results from executing action $ac_i$ in $s_i$. Additionally we denote by $x(i,\ldots,j)$ the projection $(s_0(i,\ldots,j),ac_0(i,\ldots,j))(s_1(i,\ldots,j),ac_1(i,\ldots,j)) \ldots$ where $s_e(i,\ldots,j)$ is the projection of $s_e$ on the local states of $(T_i,\ldots,T_j)$ and \\
$ ac(i,\ldots,j) =
\left\{
	\begin{array}{ll}
		\bot  & \mbox{if } ac = a_m \mbox{ and } m \not \in \{i,\ldots,j\} \\
		\bot  & \mbox{if } ac = (snd_m,rcv_n) \mbox{ and } m,n \not \in \{i,\ldots,j\} \\
		ac & otherwise
	\end{array}
\right. 
$\\
\smartpar{Bounded Fairness.} A run $x$ of a token passing system $T^{n}_G$ is $\bfair{b}(x)$ if for every moment $m$ and every process $T_i$, $T_i$ receives the token at least once between moments $m$ and $m+b$.

\smartpar{Cutoffs for Complex Networks.}
In the presence of different network topologies, represented by the graph $G$,
we define a cutoff to be 
a bound on the size of $G$ that is sufficient to decide 
the PMCP.
Note that, in order to obtain a decision procedure for the PMCP, we not only 
need to know the size of the graphs, but also which graphs of this size we need to investigate. This is straightforward if the graph always 
falls into a simple class, such as rings, cliques, or stars, but is more 
challenging if the graph can become more complex with increasing size.

\subsection{Cutoff Results for Token Passing Systems}

Table~\ref{tab:TPS} summarizes the results of this section, generalizing 
Theorem~\ref{th-token-PromptLTL} to the case of $h$-indexed specifications. Similar to previous sections, the specifications are over states of processes. The results for local bounded fairness follow from the results for global bounded fairness. 

To prove the results of this section, we need some additional definitions.

\begin{table}[h]
\vspace{-1em}
\caption{Cutoff Results for Token Passing Systems}
\label{tab:TPS}
\scriptsize
\centering
\begin{tabular}{lccc}
\toprule
& Local Bounded Fairness~~~ & Global Bounded Fairness~~~   \\ 
\midrule
$h$-indexed \LTLmX & $2h$ & $2h$ \\
$h$-indexed \promptmX~~~ & $2h$ & $2h$ \\
\bottomrule
\end{tabular}
\vspace{-1em}
\end{table}

\smartpar{Connectivity vector \cite{Clarke04c}.}
Given two indices $i,j \in V$ in a finite directed graph~$G$, we define the connectivity vector $v(G,i,j) = (u_1,u_2,u_3,u_4,u_5,u_6)$ as follows:
\begin{itemize}
\item $u_1 = 1$ if there is a non-empty path from $i$ to $i$ that does not contain $j$. $u_1 = 0$ otherwise.% token can circulate from $i$ to $i$ without passing by $j$. $u_1 = 0$ otherwise.%there is a $j$-free path from $i$ to itself. 
\item $u_2 = 1$ if there is a path from $i$ to $j$ via vertices different from $i$ and $j$. $u_2 = 0$ otherwise.
\item $u_3 = 1$ if there is a direct edge from $i$ to $j$. $u_3 = 0$ otherwise.
\item $u_4,u_5,u_6$ are defined like $u_1,u_2,u_3$, respectively where $i$ is replaced by $j$ and vice versa.
\end{itemize}

%\smartpar{Continuous execution}. Given a run $x=x_0x_1\ldots$ of a token passing system $T^{n}_G$, we say that a process $T_i$ executes actions \emph{continuously} in $x$ between moments $t_1$ and $t_2$ iff none of the other processes have executed any actions between these two moments.

\smartpar{Immediately Sends.} Given a token passing process $T$, we fix two 
local states $q^{snd}$ and $q^{rcv}$, such that there is (i) a local path 
$q^\init, \ldots, q^{rcv}$ where $q^\init \in \initt \cap (\overline{\qt} \times \{0\})$, 
(ii) a local path $q^{rcv}, \ldots, q^{snd}$ that starts with a receive 
action, and (iii) a local path $q^{snd},\ldots, q^{rcv}$ that starts with a 
send action. 

When constructing a local run for a process $T_i$ that is currently in local state 
$q^{rcv}$, we say that $T_i$ \emph{immediately sends the token} if and only if:
\begin{enumerate}
\item $T_i$ executes consecutively all the actions on a simple path $q^{rcv}, 
\ldots, q^{snd}$, then sends the token, and then executes consecutively all 
the actions on a simple path $q^{snd},\ldots, q^{rcv}$.
\item All other processes remain idle until $T_i$ reaches $q^{rcv}$.
\end{enumerate}
Note that, when $T_i$ \emph{immediately} sends the token, it executes at most 
$|\qt|$ actions, since the two paths cannot share any states except $q^{rcv}$ and $q^{snd}$.

\begin{theorem}[Cutoff for \PromptmX]\label{th-token-PromptLTL}
Let $T^{n}_G$ be a token-passing system, $g,h \in V$, and $\varphi(T_g, T_h)$ a specification with $\varphi \in \PromptmX$. Then there exists a system $T^{4}_{G'}$ with $G'=(V',E')$ and $i,j \in V'$ such that $v(G,g,h) = v(G',i,j)$, and

\begin{center}
$T^{n}_G$ $\gbmodels$ $\varphi(T_g, T_h) \iff T^{4}_{G'}$ $\gbmodels$ $\varphi(T_i, T_j)$.
\end{center}

\end{theorem}

We prove the theorem by proving two lemmas, one for each direction of the equivalence. Note that $T^{n}_G$ $\ngbmodels$ $\varphi(T_g, T_h)$ iff $\exists b \in \Nat~ \forall k \in \Nat~ \exists x \in T^{n}_G:~ \bfair{b}(x) \land (x,0,k) \not \models \varphi(T_g, T_h).$\\

\begin{lemma}[Monotonicity Lemma]\label{lemma-token-mono}
%let $n \geq 4$, then given two token passing systems $T^{n}_G$, $T^{n+1}_{G'}$, and four indices $g,h,i,j \in \Nat$ where $v(G,g,h) = v(G',i,j)$ we have:\ms{We do not need to assume anything on the structure of $G$ and $G'$ and the additional process as long as we have $v(G,g,h) = v(G',i,j)$ . check!} \ms{why 4? I think 3 would work also}
Let $T^{n}_G$ be a token-passing system with $n \geq 3$ and $g,h \in V$, and $\varphi(T_g, T_h)$ a specification with $\varphi \in \PromptmX$. Then there exists a system $T^{n+1}_{G'}$ with $G'=(V',E')$ and $i,j \in V'$ such that $v(G,g,h) = v(G',i,j)$ and 
$$T^{n}_G \ngbmodels \varphi(T_g, T_h)  ~\Rightarrow~ T^{n+1}_{G'} \ngbmodels \varphi(T_i, T_j).$$
\end{lemma}

\begin{proof}
Let $a$ be a vertex of $G$ with $a \not \in \{g,h\}$. Then we construct $G'$ from $G$ as follows: Let $V' = V \cup \{n+1\}$, and $E' = (E \cup \{(n+1, m) | (a,m) \in E $ for some $m \in V \} \cup \{(a,n+1)\}) \setminus \{(a, m) | (a,m) \in E $ for some $m \in V\}$, i.e. we copy all the outgoing edges of $a$ to the vertex $n+1$, and replace all the outgoing edges of $a$ by one outgoing edge to $n+1$.

Assume $T^{n}_G$ $\ngbmodels$ $\varphi(T_g, T_h)$. Then there exists $b \in 
\Nat$ such that $\forall k' \in \Nat$ there is a run $x$ of $T^{n}_G$ where 
$\bfair{b}(x)$, and $(x,0,|\qt| \cdot k') \not \models \varphi(T_g, T_h)$. Let 
$b' = b + (b - n + 2) \cdot |\qt|$, and $d = |Q_T| + 1$. We will construct 
for every such run $x$ a run $y$ of $T^{n+1}_{G'}$ where $\bfair{b'}(y)$, and 
$x(T_g, T_h) \equiv_d y(T_i,T_j)$ which guarantees that $(y, 0, k') \not 
\models \varphi(T_i,T_j)$ (see Corollary \ref{cor-stutter}).

\smartpar{Construction.} 
The construction is such that we keep the local paths of the $n$ existing 
processes up to bounded stuttering, and we add a process $T_{n+1}$ that 
always immediately sends the token after receiving it, with $q^{rcv}, q^{snd}$
 and the corresponding paths as defined above. In the following, as a short-hand notation, if $s=(q_1,\ldots,q_n)$ is a global state of $T^n_G$ and $q \in \qt$, we write $(s,q)$ for $(q_1,\ldots,q_n,q)$.

Let $x=(s_0,ac_0)(s_1,ac_1)\ldots$ and $y'=((s_0,q^{rcv}),ac_0)((s_1,q^{rcv})
,ac_1)\ldots$. Note that $y'$ is a sequence of configurations of $T^{n+1}_{G'}$, 
but not a run. To obtain a run, first let $y''=((s_0,q^\init),\epsilon)
\ldots((s_0,q^{rcv}),ac_0)((s_1,q^{rcv}),ac_1)\ldots$.
Finally, replace every occurrence of a pair of consecutive configurations 
$((s,q^{rcv}),(snd_a,rcv_z)),$ $((s',q^{rcv}),ac')$, where $s,s' \in \qt^n, z \in V, ac' \in \Sigma$, with the sequence\\ 
$((s,q^{rcv}),(snd_a,rcv_{n+1}))\ldots((s,q^{snd}),(snd_{n+1},rcv_z))\ldots((s',q^{rcv}),ac')$. 

In other 
words, instead of sending the token to $T_z$, $T_a$ sends the token to $T_{n+1
}$, and $T_{n+1}$ sends the token immediately to $T_z$.
Furthermore, in $x$ between moments $t$ and $t + b$, $T_a$ can send the token 
at most $b-n+1$ times, and whenever $T_{n+1}$ receives the token, it takes at most $|Q_T|$ steps before reaching $q^{rcv}$ again. 
Finally, note that the number of steps $T_{n+1}$ takes to reach $q^{rcv}$ for the first time is also bounded by $|Q_T|$.
Therefore we have $\bfair{b'}(y)$ 
and $x(T_g,T_h) \equiv_{d} y(T_i,T_j)$ (as $b'\leq b \cdot d$) which by Corollary \ref{cor-stutter} implies that $(y, 0, k') \not \models \varphi(T_i, T_j)$.\qed

\end{proof}

%Then we construct the local run of $T_{n+1}$ and update the synchronous actions of $T_a$  in $y$ as follows: 
%In $G'$, $T_a$ can send the token only through $T_{n+1}$ and $T_{n+1}$ receives the token only from $T_a$. 
%After receiving the token from $T_a$, $T_{n+1} $ \emph{immediately} send the token executes continuously all the actions that lead to the closest state that can execute $snd$ and after sending the token it will execute continuously all the actions that lead to the closest state that can execute $rcv$. Formally, let $z \in \{1,\ldots,n\} \setminus \{a\}$, then we will replace any configuration $(s,(snd_a,rcv_z))$ in $x$ with the following: $(s',(snd_a,rcv_{n+1}))$,$(s'_1,a^1_{n+1})$,$\ldots$,$(s'_f,a^f_{n+1})$,$(s',(snd_{n+1},rcv_z))$ where $a^1_{n+1}, \ldots, a^f_{n+1}$ are asynchronous actions executed by $T_{n+1}$, and $s'(1,\ldots,n) = s$. Furthermore, in every $b$-block of $x$, $T_a$ sends the token at most $b-n+1$ times and the longest path (sequence of configuration) between two states of $T_{n+1}$ that are ready to receive the token has the size $|\qt|$ (i.e. $f \leq |Q_T|$ due to the continuous execution), therefore we have $\bfair{b'}(y)$ and $x(T_g,T_h) \equiv_{d} y(T_i,T_j)$ ($b'\leq b \cdot (|\qt|+1)$). Hence, $(y, 0, k') \not \models \varphi(T_i, T_j)$.\qed

%\end{proof}

\begin{lemma}[Bounding Lemma]\label{lemma-token-bound}
Let $T^{n}_G$ be a system with $n \geq 4$ and $g,h \in V$, and $\varphi(T_g, T_h)$ a specification with $\varphi \in \PromptmX$. Then there exists a system $T^{4}_{G'}$ with $G'=(V',E')$ and $i,j \in V'$ such that $v(G,g,h) = v(G',i,j)$ and 
$$T^{n}_G \ngbmodels \varphi(T_g, T_h) ~\Rightarrow~ T^{4}_{G'} \ngbmodels \varphi(T_i, T_j).$$
\end{lemma}

\begin{proof}[Proof idea, for full proof see Appendix~\ref{app:proof-token-bound}]
First, note that the existence of $G'$ and $i,j \in V'$ with $v(G,g,h) = v(G',i,j)$ follows directly from Proposition 1 in Clarke et al.~\cite{Clarke04c}. As usual, assuming that $T^{n}_G \ngbmodels \varphi(T_g, T_h)$, we need to construct counterexample runs of $T^{4}_{G'}$ for some $b' \in \Nat$ and all $k' \in \Nat$.

The construction is based on the same ideas as in the proof of Lemma~\ref{lemma-token-mono}, with the following modifications:
i) instead of keeping all local runs of a run $x \in T^{n}_G$, we only keep the local runs of $T_g$ and $T_h$ (now assigned to $T_i$ and $T_j$), ii) instead of constructing one local run for the new process, we now construct local runs for two new processes $T_k$ and $T_l$ (basically, each of them is responsible for passing the token to $T_i$ or $T_j$, respectively), and iii) the details of the construction of these runs depend on the connectivity vector $v(G,g,h)$, which essentially determines which of the new processes holds the token when neither $T_i$ nor $T_j$ have it.

As usual, the construction ensures that $y$ is globally bounded fair and that $y(T_i,T_j) \equiv_d x(T_g,T_h)$ for some $d$, which by Corollary \ref{cor-stutter} implies that $(y, 0, k') \not \models \varphi(T_i, T_j)$.
\qed
\end{proof}

%\begin{cor}[Fairness Cutoff and Promptness bound]

%Given two token passing systems $T^{n}_G$, $T^{4}_{G'}$, and four indices $g,h,i,j$ where $v(G,g,h) = v(G',i,j)$ and $T^{4}_{G'}$ $\gbmodels \varphi(T_i, T_j)$ then $\forall n \geq 4$:

%\begin{center}
%\textbf{if} $\forall x' \in T^{4}_{G'}:~ b'\_fair(x') \Rightarrow (x',0,k') \models \varphi(T_i, T_j).$ \\
%\textbf{then} $\forall x \in T^{n}_G:~ \bfair(x) \Rightarrow (x,0,k) \models \varphi(T_i, T_j).$

%\end{center}
%with $k= k' \cdot |\qt|$, $b = b'/|\qt|$ \ms{something wrong!}
%\end{cor}

\section{Conclusions}

We have investigated the behavior of concurrent systems with respect to 
promptness properties specified in \PromptmX. Our first important observation 
is that \PromptmX is not stutter insensitive, so the standard notion of 
stutter equivalence is insufficient to compare traces of concurrent systems if we are interested in 
promptness. Based on this, we have defined \emph{bounded stutter 
equivalence}, and have shown that \PromptmX is \emph{bounded stutter 
insensitive}.

We have shown how this allows us to obtain cutoff results for guarded 
protocols and token-passing systems, and have obtained cutoffs for \PromptmX (with locally or globally bounded fairness) that are the same as those that were previously shown for \LTLmX (with unbounded fairness). This implies that, for the cases where we do obtain cutoffs, the PMCP for \PromptmX has the same asymptotic complexity as the PMCP for \LTLmX.

One case that we investigated 
remains open: disjunctive systems with global bounded fairness.
In future work, we will try to solve this 
open problem, and investigate whether other cutoff results in the 
literature can also be lifted from \LTLmX to \PromptmX.
%final version: mention distributed and parameterized synthesis for prompt-ltl

Finally, we note that together with methods for distributed synthesis from 
\PromptmX specifications, our cutoff results enable the synthesis of 
parameterized systems based on the \emph{parameterized synthesis} 
approach~\cite{Jacobs14} that has been used to solve challenging synthesis 
benchmarks by reducing them to systems with a small number of 
components~\cite{Khalimov13,BloemJK14}.

\bibliographystyle{splncs04}
\bibliography{paper}
\newpage
\appendix 
\section{Full Proof of Theorem ~\ref{theorem-stutter}}
\label{app:proof1}

Let $w, w'$ be $d$-stutter equivalent words, $\varphi$ a \PromptmX formula $\varphi$, and $f$ as defined above. Then $\forall i,k \in \Nat$:
$$ \text{if } (w,i,k) \models \varphi \text{ then } \forall j \in f(i): (w',j,d \cdot k) \models \varphi.$$

\begin{proof}

The proof works inductively over the structure of $\varphi$. Let $w_0,w_1,w_2,\ldots$ and $w'_0,w'_1,w'_2,\ldots$ be two d-compatible sequences of $w$ and $w'$. We denote by $n_i,m_i$ the number of elements inside $N_i^{w},N_i^{w'}$ respectively.\\

\textbf{Case 1: $\varphi = a$.}
$(w,i,k) \models \varphi$ $\iff$ $a \in w(i)$. By definition of $f$ we have  $\forall j \in f(i):$ $w(i) = w'(j)$, and thus $\forall j \in f(i):$ $(w', j, d\cdot k) \models \varphi$.\\

\textbf{Case 2: $\varphi = \neg a$.}
$(w,i,k) \models \varphi$ $\iff$ $a \not\in w(i)$. By definition of $f$ we have  $\forall j \in f(i):$ $w(i) = w'(j)$, and thus $\forall j \in f(i):$ $(w',j,d \cdot k) \models \varphi$.\\

\textbf{Case 3: $\varphi = \varphi_1 * \varphi_2$ with $* \in \{ \land, \lor\}$.}
$(w,i,k) \models \varphi$ $\iff$ $(w,i,k) \models \varphi_1$ $*$ $(w,i,k) \models \varphi_2$. By induction hypothesis we have: $\forall j \in f(i)$ $(w',j,d \cdot k) \models \varphi_1$ $*$ $\forall j \in f(0)$ $(w',j,d \cdot k) \models \varphi_2$ $\iff$  $(w',j,d \cdot k) \models \varphi$.\\

\textbf{Case 4: $\varphi = \F_{\p} \varphi$.}
$(w,i,k) \models \F_{\p} \varphi$ $\iff \exists e,x: i \leq e \leq i+ k,$ $e \in N_x^{w}$, and $(w,e,k) \models \varphi$ where $(\sum_{l=0}^{x-1} n_l) \leq e < (\sum_{l=0}^{x} n_l)$. Then by induction hypothesis we have: $\forall j \in f(e)$ $(w',j,d \cdot k) \models \varphi$. Let $s$ be the smallest position in $f(e)$, then $s = \sum_{l=0}^{x-1} m_l$. There exists $y \in \Nat$ s.t. $i \in N_y^w$ then $s = \sum_{l=0}^{y-1} m_l + \sum_{l=y}^{x-1} m_l$ $\leq \sum_{l=0}^{y-1} m_l + \sum_{l=y}^{x-1} n_l.d \leq \sum_{l=0}^{y-1} m_l + d.(\sum_{l=y}^{x-1} n_l) $ $\leq \sum_{l=0}^{y-1} m_l + k \cdot d$ (note that $i \in N_y^w$ and $(w,i,k) \models \F_{\p} \varphi$). As $\sum_{l=0}^{y-1} m_l$ is the smallest position in $f(i)$, then $\forall j \in f(i):$ $(w',j,d \cdot k) \models \F_{\p} \varphi$.\\%However $\forall l$ $m_l \leq n_l.d$.

\textbf{Case 5: $\varphi = \varphi_1 \U \varphi_2$.}
$(w,i,k) \models \varphi_1 \U \varphi_2$ $\iff \exists j \geq i:$ $(w,j,k) \models \varphi_2$ and $\forall e < j:$ $(w,e,k) \models \varphi_1$. Then, by induction hypothesis we have: $\forall e < j$ $\forall l \in f(e): (w', l, d \cdot k) \models \varphi_1$ and $\forall l \in f(j): (w', l, d \cdot k) \models \varphi_2$, therefore $\forall j \in f(i):$ $(w',j,d \cdot k) \models \varphi_1 \U \varphi_2$\\

\textbf{Case 6: $\varphi = \varphi_1 \R \varphi_2$.}
$(w,i,k) \models \varphi$ then either $\forall e \geq i$ $(w, e, k) \models \varphi_2$ or $\exists e \geq i: (w, e, k) \models \varphi_1 \land \forall j \leq e$ $(w, j, k) \models \varphi_2$

\begin{itemize}

\item \textbf{Subcase:} $\forall e \geq i$ $(w, e, k) \models \varphi_2$. By induction hypothesis we have $\forall e \geq i$ $\forall j \in f(e):$ $(w',j,d \cdot k) \models \varphi_2 $ then $\forall j \in f(i):$ $(w',j,d \cdot k) \models \varphi$

\item \textbf{Subcase:} $\exists e \geq i: (w, e, k) \models \varphi_1 \land \forall j \leq e$ $(w, j, k) \models \varphi_2$. Then, by induction hypothesis , we have: $\forall l \in f(e): (w', l, d \cdot k) \models \varphi_1$ and $\forall j \leq e$ $\forall l \in f(e): (w', l, d \cdot k) \models \varphi_2$, therefore $\forall j \in f(0):$ $(w',j,d \cdot k) \models \varphi$
\end{itemize}
%\ms{
%\textbf{Case 7:} $\varphi = G_p \varphi$.
%$(w,i,k') \models G_p \varphi$ $\iff \forall e: i \leq e \leq k'$ and $(w,e,k') \models \varphi$. Then by induction hypothesis and lemma-stutter, we have: $\forall e$ $\forall j \in f(e)$ $(w',j,t) \models \varphi$. Then $\forall j \in f(i):$ $(w',j,k' * \frac{1}{d}) \models G_p \varphi$.
%}
\qed
\end{proof}

\section{The absence of a bounding lemma under global bounded fairness.}
\label{app:absence}
The reader will notice that we have no bounding lemma, and therefore no cutoff result, for \PromptmX under global  bounded
fairness. The main reason is that the constructions we 
adopt do not allow us to determine a bound on the number of stuttering steps they 
generate. For instance, the proof of Lemma~\ref{bound-disj-LTL-b-fair-1} depends on a bound on the time after which only infinitely visited states occur. Based on the existence of an ultimately periodic counterexample $uv^{\omega}$, we can conclude that $|u|$ is sufficient as a bound. In case of \PromptmX however, this technique is not sufficient: a \Prompt counterexample consists of a fairness bound $b$ such that for all $k$ there is a non-satisfying run. Since the previously mentioned technique only produces a bound $b$ that will depend on the run for a given $k$, it cannot solve our problem. 

As an alternative approach, we tried a technique based on the algorithm for solving the model checking problem for \Prompt by Kupferman et al.~\cite{KupfermanPV09}. Their method is based on the detection of a \emph{pumpable path} in the product of a system $S$ and a specification automaton $\buchi{\varphi}$. However, when constructing a pumpable path for $\largesyse{c}$ out of a pumpable path of $\largesys$, we run into the problem that in certain cases the value of $c$ depends on $n$, and therefore no cutoff can be detected with this technique.

\section{Full Proof of Lemma~\ref{periodic-CE}}
\label{app:proof4}

A \emph{run graph} of a B\"uchi automaton $\buchi{\varphi} = (Q_A \times Q_B^n ,Q_{\buchi{\varphi}}, \delta, 
a_0, \alpha)$ on a system $\fairSys{b}{n}$ is a directed graph $\mG^{n}_b(\varphi) = (V,E)$ where:
\begin{itemize}
\item $V \subseteq (Q_A \times Q_B^n) \times \{0,\ldots,b\}^{n+1} \times Q_{\buchi{\varphi}}$
\item $(s_0,b_0,a_0) \in V$, where $b_0$ denotes that all counters are set to $0$. 
\item $((s,b,a),(s',b',a')) \in E$ iff $(s,s') \in \Delta,$
 $a' \in \delta(a, s)$, and $b'$ results from $b$ according to the rules for the counters.
\end{itemize}
An infinite path of the run graph $\pi = (s_0,b_0,a_0)(s_1,b_1,a_1)\ldots$ is an \emph{accepting path} if it starts with $(s_0,b_0,a_0)$, and visits a state $a_{\alpha} \in \alpha$ infinitely often.

Lemma \ref{periodic-CE}: Let $\varphi \in \LTL$ and $b \in \Nat$. If 
$\fairSys{b}{n} \not \models \varphi$ then there exists a run $x=uv^{\omega}$ 
of $\largesys$ with $\bfair{b}(x)$, and $x \not \models \varphi$, where $u,v$ 
are finite paths, and $|u|, |v| \leq 2 \cdot |Q_A| \cdot |Q_B|^n \cdot b^{n+1}
 \cdot |Q_{\buchi{\neg \varphi}}|$.

\begin{proof}
Assume that $\fairSys{b}{n} \not \models 
\varphi$. Then there exists an accepting path $\pi'$ in the run graph 
$\mG^{n}_b(\neg \varphi)$. We first construct out of $\pi'$ a fair path $\pi = u_{\pi}
v_{\pi}^{\omega}$, by detecting and extracting a lasso-shaped accepting path 
from $\pi'$. In $\pi'$ there exists an infix $\pi'_i \ldots \pi'_j$ where 
$\pi'_i = \pi'_j$, and there exists $\pi'_l \in \{\pi'_{i+1}, \ldots, \pi'_{j-1}\}$
with $\pi'_l(Q_{\buchi{\neg \varphi}}) \in \alpha$ (accepting state in the 
automaton). Therefore 
$\pi'_0 \ldots \pi'_{i-1}(\pi'_i \ldots \pi'_{j-1})^{\omega}$ 
is an accepting path of $\mG^{n}_b(\neg \varphi)$. 

Let $u' = \pi'_0 
\ldots \pi'_{i-1}$ and $v' = \pi'_i \ldots \pi'_{j-1}$, then we can construct 
$u_{\pi}$ and $v_{\pi}$ by detection and removal of cycles under some 
conditions: (i) let $u_{\pi}$ be a finite path obtained form $u'$ where we 
iteratively replace every infix $\pi'_s \ldots \pi'_t$ with $\pi'_s$ if 
$\pi'_s = \pi'_t$. Then, since $u_\pi$ does not contain repetitions, we have 
$u_{\pi} \leq |Q_A| \cdot |Q_B|^n \cdot b^{n+1} \cdot 
|Q_{\buchi{\neg \varphi}}|$.
(ii) let $\pi'_a \in \{\pi'_i \ldots \pi'_{j-1}\}$ where $\pi'_a(\buchi{\neg 
\varphi}) \in \alpha$ and let $v_{\pi}$ be a finite path obtained form $v'$ 
after we iteratively replace every infix $\pi'_s \ldots \pi'_t$ with $\pi'_s$ 
if $\pi'_s = \pi'_t$ and $s \geq a$ or $t < a$. Thus, we get $v_{\pi} \leq 2 \cdot 
|Q_A| \cdot |Q_B|^n \cdot b^{n+1} \cdot |Q_{\buchi{\neg \varphi}}|$.

Finally, let $x=u_{\pi}(Q_A \times Q_B^n)\left(v_{\pi}(Q_A \times Q_B^n)\right)^{\omega}$. 
By construction, $x$ is a run of $\largesys$ with $\bfair{b}(x)$ and $x \not \models \varphi$.
\qed
\end{proof}

%\section{Full Proof of Lemma~\ref{}}
%\label{app:proof-mon-con}
%
%Assume $\largesys$ $\nlbmodels$ $\varphi(A, B^{(1)})$. Then there exists $b 
%\in \Nat$ such that $\forall k \in \Nat$ there is a run $x$ of $\largesys$ 
%where $\bfair{b}(x)$ and $(x,0, k) \not \models \varphi(A, B^{(1)})$. For 
%every such $x$, we construct a run $y$ of $\largesyse{n+1}$ with $\lbfair
%{b}(y)$ and $(y,0,k) \not \models \varphi(A, B^{(1)})$.
%Let $y(A) = x(A)$ and $y(B_j) = x(B_j)$ for all $B_j \in \{B_1,\ldots,B_n\}$ 
%and let the new process $B_{n+1}$ "share" a local run $x(B_i)$ with an 
%existing process $B_i$ of $\largesyse{n+1}$ in the following way: one process stutters in $init_B$ 
%while the other makes transitions from $x(B_i)$, and whenever $x(B_i)$ enters 
%$init_B$ the roles are reversed. Since this changes the behavior of $B_i$, 
%$B_i$ cannot be a process that is mentioned in the formula, i.e. we need $n \geq 2$ for a 
%formula $\varphi(A,B^{(1)})$. Then we have $\lbfair{b}(y,\{A,B_1\})$ as the 
%run of $B_{n+1}$ inherits the unconditional fairness behavior from the local 
%run of the process $B_i$ in $x$. Note that it is not guaranteed that the 
%local runs $y(B_i)$ and $y(B_{n+1})$ are bounded fair as the time between two 
%occurrences of $init_{B}$ in $x(B_i)$ is not bounded. Moreover we have 
%$x(A,B_1) \equiv_{1} y(A, B_1)$, which according to Corollary~\ref{cor-stutter} implies
%$(y(A, B_1), k) \not \models  \varphi(A, B^{(1)})$. 
%\qed

\section{Full Proof of Lemma~\ref{lemma-token-bound}}
\label{app:proof-token-bound}

Let $T^{n}_G$ be a system with $n \geq 4$ and $g,h \in V$, and $\varphi(T_g, T_h)$ a specification with $\varphi \in \PromptmX$. Then there exists a system $T^{4}_{G'}$ with $G'=(V',E')$ and $i,j \in V'$ such that $v(G,g,h) = v(G',i,j)$ and 
$$T^{n}_G \ngbmodels \varphi(T_g, T_h) ~\Rightarrow~ T^{4}_{G'} \ngbmodels \varphi(T_i, T_j).$$

\begin{proof}

Let $i$, $j$, $k$, and $l$ be the processes indices in $T^{4}_{G'}$. $G'=(V',E')$ is any graph where $V'=\{i,j,k,l\}$, $v(G,g,h) = v(G',i,j)$, and  $(k,j),(l,i) \in E'$. According to \cite{Clarke04c} such graph always exists.

Assume $T^{n}_G \ngbmodels \varphi(T_g, T_h)$. Then there 
exists $b \in \Nat$ such that $\forall k' \in \Nat$ there is a run $x$ of 
$T^{n}_G$ where $\bfair{b}(x)$, and $(x,0,k' \cdot (|\qt|+1)) \not \models 
\varphi(T_g, T_h)$. Let $d = |\qt| + 1$, and $b' = 2|\qt| + b + (b-n+2) \cdot |\qt|$. We show how to construct for every such $x$ a run $y$ of $T^{4}_{G'}$ where $\bfair{b'}(y)$, $x(T_g, T_h) \equiv_{d} y(T_i, T_j)$. \\
\smartpar{Construction.}
Let $x=(s_0,ac_0)(s_1,ac_1)\ldots$ and 
\[y'=((s_0(T_g, T_h),q^{rcv},q^{rcv}),ac_0(T_g, T_h))((s_1(T_g, T_h),q^{rcv},q^{rcv}),ac_1(T_g, T_h))\ldots\]
The word $y'$ is a sequence of configurations of $T^{4}_{G'}$, where we assign the local runs of $T_g, T_h$ into the local runs of $T_i$  and $T_j$.
Note that $y'$ is not a run, hence to obtain a run, first let 
\begin{gather*}
\begin{split}
y''=((s_0(T_g, T_h),q^\init,q^\init),\epsilon)\ldots((s_0(T_g, T_h),q^{rcv},q^{rcv}),ac_0(T_g, T_h))\\
((s_1(T_g, T_h),q^{rcv},q^{rcv}),ac_1(T_g, T_h))\ldots
\end{split}   
\end{gather*}

If neither $T_g$ nor $T_h$ has the token in the initial state of $x$, then, if $T_g$ has the token first in $x$ before $T_h$, we replace the pair of consecutive configurations\\
\[((s(T_g,T_h),q^{rcv},q^{rcv}),(snd_z,rcv_i))((s'(T_g,T_h),q^{rcv},q^{rcv}),ac'(T_g,T_h))\] with 
\begin{gather*}
\begin{split}
((s(T_g,T_h),q^{rcv},q^{rcv}),\epsilon)\ldots((s(T_g,T_h),q^{snd},q^{rcv}),(snd_i,rcv_i))\\ \ldots((s'(T_g,T_h),q^{rcv},q^{rcv}),ac'(T_g,T_h))
\end{split}   
\end{gather*}
where $z \in V$.
Similarly we deal with the case where $T_h$ has the token before $T_g$.\\
Furthermore, for every occurrence of a pair of consecutive configurations
\[pair_i=((s(T_g,T_h),q^{rcv},q^{rcv}),(snd_i,rcv_z))((s'(T_g,T_h),q^{rcv},q^{rcv}),ac'(T_g,T_h))\]
where $s,s' \in \qt^n, z \in V \setminus \{j\}, ac' \in \Sigma$, then:
\begin{itemize}
\item If after $pair_i$ in $y''$ $T_i$ executes the receive action without a receive action from $T_j$ in between, then $(i,l),(l,i) \in E'$, and we replace $pair_i$ with the sequence:
\begin{gather*}
\begin{split}
((s(T_g,T_h),q^{rcv},q^{rcv}),(snd_i,rcv_l))\ldots((s(T_g,T_h),q^{snd},q^{rcv}),(snd_l,rcv_i))\\\ldots((s'(T_g,T_h),q^{rcv},q^{rcv}),ac'(T_g,T_h))
\end{split}   
\end{gather*}
Informally we let the process $T_l$ receive the token from $T_i$ and send it immediately back to $T_i$.
\item If after $pair_i$ in $y''$ $T_j$ receives the token through some other process(es) (different than $T_i$ and $T_j$), then $(i,k),(k,j) \in E'$, and we replace $pair_i$ with the sequence:
\begin{gather*}
\begin{split}
((s(T_g,T_h),q^{rcv},q^{rcv}),(snd_i,rcv_k))\ldots((s(T_g,T_h),q^{rcv},q^{snd}),(snd_k,rcv_j))\\\ldots((s'(T_g,T_h),q^{rcv},q^{rcv}),ac'(T_g,T_h))
\end{split}   
\end{gather*}
Informally we let the process $T_k$ receive the token from $T_i$ and sends \emph{immediately} back to $T_j$.
\end{itemize}
Next, we do the same for every occurrence of a pair of consecutive configurations \\
\[pair_j=((s(T_g,T_h),q^{rcv},q^{rcv}),(snd_j,rcv_z))((s'(T_g,T_h),q^{rcv},q^{rcv}),ac'(T_g,T_h))\]
where $s,s' \in \qt^n, z \in V \setminus \{i\}, ac' \in \Sigma$.

Furthermore, in $x$ between moments $t$ and $t + b$, $T_g$ and $T_h$ can send the token 
at most $b-n+2$ times, and whenever $T_l$ or $T_k$ receives the token, it takes at most $|Q_T|$ steps before reaching $q^{rcv}$ again.
Finally, note that the number of steps $T_l$ or $T_k$ takes to reach $q^{rcv}$ for the first time is also bounded by $|Q_T|$.
Therefore we have $\bfair{b'}(y)$ 
and $x(T_g,T_h) \equiv_{d} y(T_i,T_j)$ ($b'\leq b \cdot d$)  which by Corollary \ref{cor-stutter} implies that $(y, 0, k') \not \models \varphi(T_i, T_j)$.\qed
\end{proof}

\end{document}